\documentclass[11pt]{amsart}
\usepackage{amssymb,amsmath,amsthm}
\usepackage{mathrsfs,a4wide}
\usepackage{esint}
\usepackage{epsfig}
\usepackage{appendix}
\newtheorem{theorem}{Theorem}[section]
\newtheorem{lemma}[theorem]{Lemma}
\newtheorem{proposition}[theorem]{Proposition}
\newtheorem{corollary}[theorem]{Corollary}

\newtheorem{definition}{Definition}
\theoremstyle{definition}

\numberwithin{equation}{section}

\usepackage[T1]{fontenc}  
\usepackage{hyperref}

\theoremstyle{definition}
\newtheorem{ciao}[theorem]{Remark}
\newtheorem{problem}[theorem]{Problem}

\newcommand{\N}{\mathbb{N}}
\newcommand{\Z}{\mathbb{Z}}

\newcommand{\T}{\mathbb{T}}
\newcommand{\R}{\mathbb{R}}
\newcommand{\C}{\mathbb{C}}

\newcommand{\1}{\mathbf{1}}  
\newcommand{\Id}{\mathbb{I}} 

\newcommand{\I}{\mathrm{i}}

\newcommand{\Hi}{{\mathcal{H}}}
\newcommand{\Ht}{{\mathcal{H}_{\tau}}}
\newcommand{\Hf}{\mathcal{H}_{\mathrm{f}}}

\newcommand{\B}{\mathcal{B}}
\newcommand{\Un}{\mathcal{U}}
\newcommand{\UB}{\mathcal{U}_{\rm BF}}
\newcommand{\UZ}{\tilde{\mathcal{U}}_{\rm BF}}

\newcommand{\J}{\mathcal{J}}
\newcommand{\Do}{\mathcal{D}}
\newcommand{\Sch}{\mathcal{S}}

\newcommand{\W}{\mathcal{W}}

\newcommand{\E}{{\mathcal E}}

\newcommand{\epsi}{\varepsilon}
\newcommand{\ph}{\varphi}
\newcommand{\la}{\lambda}
\newcommand{\ta}{\tau}

\newcommand{\Base}{\T_d^*}
\newcommand{\SU}{\mathcal{SU}(m)}
\newcommand{\su}{\mathfrak{su}(m)}

\newcommand{\indexd}{\in\{ 1,\ldots,d \}}

\newcommand{\virg}[1]{\lq\lq #1 \rq\rq}                
\newcommand{\ie}{{\sl i.e.\ }}
\newcommand{\eg}{{\sl e.g.\ }}

\newcommand{\expo}[1]{\mathrm{e}^{#1}}     
\newcommand{\inner}[2]{\left \langle  #1 \, , \,  #2 \right \rangle}
\newcommand{\norm}[1]{\| #1 \|}
\newcommand{\normHS}[1]{\| #1 \|}   

\newcommand{\transpose}[1]{ {#1}^{\intercal}}

\newcommand{\de}{{\partial}}
\newcommand{\set}[1]{ \left\{  #1 \right\}}
\newcommand{\inset}[1]{ \in \left \{  #1 \right \}}
\newcommand{\Vol}{ \mathrm{Vol}}

\def\({\left(}
\def\){\right)}
\def\<{\left\langle}
\def\>{\right\rangle}

\newcommand{\half}{\frac{1}{2}}

\DeclareMathOperator{\Tr}{Tr}                            
\DeclareMathOperator{\tr}{tr}                            
\DeclareMathOperator{\re}{Re} 
\DeclareMathOperator{\Ran}{Ran} \DeclareMathOperator{\ran}{Ran}
\DeclareMathOperator{\Span}{Span} \DeclareMathOperator{\diag}{diag}
\DeclareMathOperator{\Lap}{\Delta} \DeclareMathOperator{\dive}{div}
 \DeclareMathOperator{\Var}{Var}
\DeclareMathOperator{\curl}{curl}
\DeclareMathOperator{\supp}{supp}

\usepackage{color}

\newcommand{\avint}{\fint}

\title
[Bloch bundles, Marzari-Vanderbilt functional and Wannier functions]
{Bloch bundles, Marzari-Vanderbilt functional and maximally
localized Wannier functions}
\author[G. Panati]
{Gianluca Panati}
\author[A. Pisante]
{Adriano Pisante}
\address[G. Panati]{Dipartimento di Matematica "G. Castelnuovo",  "La Sapienza" Università di Roma,
Piazzale A. Moro 2, 00185 Roma, Italy} \email{panati@mat.uniroma1.it}

\address[A. Pisante]{Dipartimento di Matematica "G. Castelnuovo", "La Sapienza" Università di Roma,
Piazzale A. Moro 2, 00185 Roma, Italy}
\email{pisante@mat.uniroma1.it}

\begin{document}
\vskip .2truecm

\begin{abstract}
We consider a periodic Schr\"{o}dinger operator and the composite
Wannier functions corresponding to a relevant family of its Bloch
bands, separated by a gap from the rest of the spectrum. We study
the associated localization functional introduced in \cite{MaVa} and
we prove some results about the existence and exponential
localization of its minimizers, in dimension $d \leq 3$. The proof
exploits ideas and methods from the theory of harmonic maps between
Riemannian manifolds.

\end{abstract}
\maketitle

\vskip -.5truecm
{\small \tableofcontents}

\section{Introduction}

Many transport properties of electrons in crystalline solids are
understood by the analysis of Schr\"{o}dinger operators in the form
\begin{equation}\label{Hamiltonian}
    H = -   \Delta +  V_{\Gamma}(x) \qquad \mbox{acting in }
    L^2(\R^d)
\end{equation}
where the function $V_{\Gamma}: \R^d \rightarrow \R$  is periodic
with respect to a Bravais lattice $\Gamma \simeq \Z^d$. The function
$V_{\Gamma}$ represents (in Rydberg units) the electrostatic
potential experienced by a test electron and generated by the ionic
cores of the crystalline solid and, in a mean-field approximation,
by the remaining electrons. We refer to \cite{CattoLeBrisLions,
CancesDelaurenceLewin} for the mathematical justification of such a
model in the reduced Hartree-Fock approximation.

A crucial problem in solid state physics is the construction of an
orthonormal basis, canonically associated to the operator $H$, consisting of
functions which are exponentially localized in space. Indeed, such a basis allows to
develop computational methods which scale linearly with the system
size \cite{Goedecker}, makes possible the description of the
dynamics by \emph{tight-binding} effective Hamiltonians \cite{NaMa,
HeSj_Scroedinger}, and plays a prominent role in the modern theories
of macroscopic polarization \cite{KSV,Re92} and of orbital
magnetization in crystalline solids \cite{ThoCerVanRes2005}.

A convenient basis has been proposed by Wannier \cite{Wannier},
and \emph{Wannier functions}  (see Definition \ref{Def Wannier}) are nowadays a fundamental tool in
solid state physics. The problem of proving the existence of
exponentially localized Wannier functions was raised in 1959 by the
Nobel Prize winner W. Kohn \cite{Kohn64}, who solved it in dimension $d=1$
in the case of a single isolated  Bloch band for a centrosymmetric
potential. The latter condition has been later removed by J. des
Cloizeaux \cite{Cl2}. In higher dimension, the problem has been
solved, always in the case of a single isolated Bloch band, by J.
des Cloizeaux \cite{Cl1, Cl2} for centrosymmetric potentials and
finally by G. Nenciu under general hypotheses \cite{Ne83}, see also
\cite{HeSj_Scroedinger} for a different proof and the review paper
\cite{Ne91}. As for dimension $d=1$, an alternative approach based
on the band position operator has been developed, yielding
exponential localization of Wannier functions even for non periodic
Schr\"odinger operators, see \cite{NeNe97, CoNe} and references
therein.

However,  in dimension $d > 1$ the Bloch bands of crystalline solids
are not, in general, isolated. Thus the interesting problem, in view
of real applications, concerns the case of \emph{multiband systems},
and in this context the more general notion of \emph{composite
Wannier functions} is relevant \cite{Bl,Cl1}. The existence of
exponentially localized composite Wannier functions has been proved
in \cite{Ne91} in dimension $d=1$.  As for $d>1$, this problem
remained unsolved until recently \cite{Pa2008, Wannier_letter_BPCM},
when an existence result was obtained by geometric methods.

To circumvent such conceptual difficulty, and in view of the
application to numerical simulations, the solid-state physics
community preferred to introduce the alternative notion of
\emph{maximally localized Wannier functions} \cite{MaVa}. The latter
are defined as the minimizers of a suitable localization functional,
known as the Marzari-Vanderbilt (MV) functional. In \cite{MaVa} it
is also conjectured that the minimizers, whenever they exist, are
exponentially localized.  While such an approach provided excellent
results from the numerical viewpoint, being nowadays an every-day
tool in computational physics \cite{Wannier review, Lyon
conference}, a mathematical analysis of the MV functional is still
missing. Our goal is to fill this gap.

In this paper, we prove preliminarily that minimizers of the MV functional do
exist for $d \leq 3$ in a suitable function space. More relevantly,
we prove that the minimizers of the MV functional are exponentially
localized. More precisely, let $m \geq 1$ be the number of Bloch
bands corresponding to the composite Wannier functions; then the
result is proved in three cases: if $m=1$ for any $d \geq 1$, if $ 1
\leq d \leq 2$ for any $m \geq 1$, and if $d=3$ under the constraint
$2 \leq m \leq 3$ (Theorem \ref{maintheorem}). The proofs are
dimension-dependent. In the first two cases, exponential
localization holds true for any stationary point of the MV
functional, \ie for any solution of the corresponding Euler Lagrange
equations, while in the three dimensional case the minimality
property is crucial. In view of that, when $d=3$ the constraint on
$m$ arises from the methods we exploit in the proof; we believe that
it is purely technical, since we do not see any physical reason
which prevents the result from being true for any $m \geq 1$. So
far, even if one is merely interested in almost-exponential
localization, \ie to prove that the maximally localized composite
Wannier functions decrease faster than the inverse of any
polynomial, the constraint on $m$ is needed in our approach.

The existence of a basis of  Wannier functions which are
exponentially localized in space is also relevant for the periodic
Pauli and Dirac operators. The $\Z^d$-symmetry of the latter
operators allows a direct integral decomposition analogous to the
Bloch-Floquet transform, yielding a structure very similar to the
one appearing for periodic Schr\"odinger operators. With few
modifications, our methods can also be used to discuss the
corresponding minimization problem in this context.


Mathematically, the MV functional can be rewritten, after
Bloch-Floquet transform and by exploiting the time-reversal symmetry
of the operator (\ref{Hamiltonian}) (\ie the fact that $H$ commutes with
the complex-conjugation operator),  as a perturbation
$\tilde{F}_{\rm MV}$ of the Dirichlet energy for maps from $\Base$
to $\Un(m)$, see (\ref{Dirichletenergy}), where $\Base \simeq
\R^d/(2 \pi \Z)^d$ is a $d$-dimensional flat torus and $\Un(m)
\subset \mathrm{M}_m(\C)$ is the unitary group. The exponential
localization of a minimizer of the MV functional is related to the
analyticity of the corresponding minimizer of $\tilde{F}_{\rm MV}$.
The existence of a minimizer of the latter functional follows
essentially  from the direct method of calculus of variations.

\noindent We prove the analyticity of the minimizers of
$\tilde{F}_{\rm MV}$ by adapting ideas and methods from the
regularity theory for harmonic maps, see \cite{CWY, LW, S} and
references therein. The crucial step is to prove that any minimizer
of $\tilde{F}_{\rm MV}$ is continuous. In the two dimensional case,
this fact is a consequence of the hidden structure of the nonlinear
terms in the Euler Lagrange equation for the $\tilde{F}_{\rm MV}$
functional. In the three dimensional case, the continuity follows
instead from the deeper fact that minimizers at smaller and smaller
scales look like minimizing harmonic maps from $\Base$ to $\Un(m)$.
We are able to prove that, for $m \leq 3$, the latter are actually
real-analytic (Theorem \ref{RegUnitons}), by showing constancy of
the tangent maps as in the important paper \cite{SU2}. As a
consequence, we obtain the continuity of the minimizers of
$\tilde{F}_{\rm MV}$ and, in turn, analytic regularity.


We hope that our result will contribute to a fruitful exchange of
ideas, problems and methods between the solid-state physics and the
mathematics community, in the study of Schr\"{o}dinger operators and
related physical phenomena.

\noindent \textbf{Acknowledgments.} We are grateful to Ch.\ Brouder,
G.\ Dell'Antonio and S.\ Teufel for a careful reading of the
manuscript and for useful comments. We are indebted to A. D'Andrea,
D. Fiorenza, G. Mondello and P. Papi for stimulating  discussions
about the geometric and Lie-theoretic aspects of the problem.


\section{Wannier functions and Bloch bundles}


\subsection{Periodic Schr\"{o}dinger operators and Bloch-Floquet transform}

In this section we recall, following \cite{PST2003}, that, in view of their invariance
under $\Gamma$-translations, Schr\"{o}dinger operators in the form \eqref{Hamiltonian} can be
decomposed as  a direct integral of simpler operators by the
(modified) Bloch-Floquet transform. A comparison with the formalism
of the classical Bloch-Floquet transform, appearing in most of the
physics literature, is summarized in Remark\,\ref{Rem physical
literature}.

The lattice $\Gamma$, corresponding to the Bravais lattice in
physics, is described as
$$
\Gamma =\Big\{ \gamma \in\R^d: \gamma = \textstyle{\sum_{j=1}^d}
n_j\,\gamma_j \,\,\,\mbox{for some }n_j \in \mathbb{Z} \Big\} \,,
$$ where $\{\gamma_1,\ldots,\gamma_d \}$ are fixed linearly independent vectors in $\R^d$. The dual
lattice, with respect to the ordinary inner product, is $ \Gamma^*
:= \set{ k \in \R^d: \ k \cdot \gamma \in 2 \pi \Z \mbox{ for all }
\gamma \in \Gamma}.$ To fix the notation, we denote by $Y$ the
centered fundamental domain of $\Gamma$, namely
$$
 Y = \Big\{ x\in\R^d: x=
\textstyle{\sum_{j=1}^d}\alpha_j\,\gamma_j
\,\,\,\mbox{for}\,\,\alpha_j \in
[-\textstyle{\frac{1}{2},\frac{1}{2}}]
 \Big\}\,.
$$
Analogously, we define the centered fundamental domain $Y^{*}$ of
$\Gamma^*$  by setting
$$
 Y^* = \Big\{ k \in\R^d:  k =
\textstyle{\sum_{j=1}^d} {k}^{\prime}_j \,\gamma^*_j
\,\,\,\mbox{for}\,\, {k}^{\prime}_j \in
[-\textstyle{\frac{1}{2},\frac{1}{2}}]
 \Big\}\,,
$$
where $\{ \gamma^*_j \}$ is the dual basis to $\set{\gamma_j}$, \ie
$\gamma^*_j \cdot \gamma_i = 2 \pi \delta_{j,i}$. When the opposite
faces of $Y^*$ are identified, one obtains the torus $\Base := \R^d
/ \Gamma^*$, equipped with the flat Riemannian metric induced by
$\R^d$.

For $\psi\in\Sch(\R^d)$, one defines the modified Bloch-Floquet
transform  as
\begin{equation} \label{Zak transform}
( \UZ \psi)(k,y):=
 \frac{1}{|Y^*|^{\half}} \,\,
\sum_{\gamma\in\Gamma} \expo{-\I k \cdot (y +\gamma)} \, \psi( y +
\gamma), \qquad y \in \R^d, \, k \in\R^{d}.
\end{equation}
One immediately reads the periodicity properties
\begin{equation}\label{Zak properties}
\begin{array}{lcll}
      \big( \UZ \psi\big) (k, y+\gamma) &=&
      \big( \UZ\psi\big)
      (k,y) &\quad \mbox{for all } \gamma \in\Gamma\,,\\
      &&&\\
      \big(\UZ\psi\big) (k + \lambda, y) &=& \expo{-\I \la \cdot y}\,\big(
      \UZ\psi\big) (k,y) &\quad \mbox{for all }\lambda\in\Gamma^*\,.
\end{array}
\end{equation}

\noindent For any fixed $k\in{\R^d}$, $\big( \UZ \psi
\big)(k,\cdot)$ is a $\Gamma$-periodic function and can thus be
regarded as an element of $\Hf := L^2(\T^d)$, $\T^d$ being the flat
torus $\R^d/\Gamma$.  On the other hand, the second equation in
(\ref{Zak properties}) can be read as a pseudoperiodicity property,
involving a unitary representation of the group $\Gamma^*$, given by
\begin{equation} \label{tau definition}
\tau:\Gamma^{*} \rightarrow \Un(\Hf)\,,\quad\lambda \mapsto
\tau(\lambda)\,, \quad \big(\tau(\lambda)\ph \big)(y) = \expo{\I \la
\cdot \,y} \ph(y).
\end{equation}
Following \cite{PST2003}, it is convenient to introduce the Hilbert
space
\begin{equation}   \label {H tau}
\Hi_\tau :=\Big\{ \ph \in L^2_{\rm loc}(\R^d, \Hf ):\,\, \ph(k -
\la) = \tau(\la)\,\ph(k) \,\,\, \forall \la \in \Gamma^{*}, \
\mbox{for a.e. } k \in \R^d \Big\}\,,
\end{equation} equipped with the inner product
\[
\langle \ph,\,\psi \rangle_{\Hi_\tau} = \int_{Y^{*}} dk\, \langle
\ph(k),\,\psi(k)\rangle_{\Hf}\,.
\]
Obviously, there is a natural isomorphism  between $\Hi_\tau$ and
$L^2(Y^{*}, \Hf)$ given by restriction from $\R^d$ to $Y^{*}$. The
map defined by (\ref{Zak transform}) extends to a unitary operator
$$
\UZ : L^2(\R^d) \longrightarrow \Hi_{\tau} \simeq
\int_{Y^*}^{\oplus} \Hf \, dk,
$$
with inverse given by
\[
\( \UZ^{-1} \ph \)(x) =  \frac{1}{|Y^*|^{\half}} \int_{Y^*} dk \,\,
\expo{ \I k \cdot x} \ph(k, [x]),
\]
where $[\, \cdot \, ]$ refers to  the a.e. unique decomposition $x =
\gamma_x + [x]$, with $\gamma_x \in \Gamma$ and $[x] \in Y$.

\noindent Finally, from the definition (\ref{Zak transform}) one
easily checks that
\begin{equation}
\label{Zak equivalence}
\begin{array}{rcl}
\psi \in W^{m,2}(\R^d), \ m \in \N & \Longleftrightarrow &   \UZ \,
\psi \in
L^2(Y^*, W^{m,2}(\T^d)), \\
&& \\
\< x \>^m \psi \in L^2(\R^d), \ m \in \N   & \Longleftrightarrow &
\UZ \, \psi \in \Ht \cap W^{m,2}_{\rm loc}(\R^d, L^2(\T^d)),\\
\end{array}
\end{equation}
\noindent where, as usual, $\< x \> = (1 + |x^2|)^{1/2}$.


\medskip

The advantage of this construction is that the transformed
Hamiltonian is a fibered operator over $Y^*$. To assure that $H =
-\Lap + V_{\Gamma}$  is self-adjoint in $L^2(\R^d)$ on the domain
$W^{2,2}(\R^d)$, we make the following Kato-type assumption on the
$\Gamma$-periodic potential \cite[Theorem XIII.96]{RS4}:
\begin{equation}\label{V assumptions}
V_{\Gamma} \in L_{\rm loc}^2(\R^d) \mbox{ for } d \leq 3, \qquad
V_{\Gamma} \in L_{\rm loc}^p(\R^d) \mbox{ with } p > d/2 \mbox{ for
} d \geq 4.
\end{equation}
One checks that
\[
\UZ \, H  \, \UZ^{-1} = \int_{Y^{*}}^\oplus dk \,H(k)
\]
with fiber operator
\begin{equation} \label{H(k)}
H(k) = \big( -\I \nabla_y + k\big)^2 + V_\Gamma(y)\,, \qquad k \in
\R^d,
\end{equation} acting  on  the
$k$-independent domain $\Do_0 = W^{2,2}(\T^d) \subset L^2(\T^d)$. Each
fiber operator $H(k)$ is self-adjoint, has compact resolvent and
thus pure point spectrum accumulating at infinity. The eigenvalues
are labeled increasingly, \ie $E_0(k) \leq E_1(k) \leq E_2(k) \leq
\ldots$, and repeated according to their multiplicity. Since the fiber Hamiltonians
are $\tau$-covariant, see
\cite{PST2003}, in the sense that
\begin{equation}\label{H covariant}
H(k + \la) = \ta(\la)^{-1} \, H(k)  \, \ta(\la), \qquad \forall \la
\in \Gamma^{*},
\end{equation}
the eigenvalues are $\Gamma^*$-periodic, \ie $E_n(k + \la) = E_n(k)$
for all $\la \in \Gamma^*$.

We denote by $\sigma_{\ast}(k)$  the set $\{ E_{i}(k): n \leq i \leq
n+m-1 \}$, $k \in Y^*$, corresponding to a physically relevant family of $m$
Bloch bands, and we assume the following \emph{gap condition}:
\begin{equation}\label{Gap condition}
    \inf_{k \in \T_d^{*}} \mathrm{dist}\big( \sigma_*(k), \sigma(H(k)) \setminus \sigma_*(k)
    \big) > 0.
\end{equation}

The following result borrows ideas in \cite{Cl1,Cl2} and
\cite[Theorem 3.3]{Ne91}, where however a different definition of
the Bloch-Floquet transform is used.


\begin{proposition}\label{Prop P properties}
Let $ P_*(k) \in \B(\Hf)$ be the spectral projector of $H(k)$
corresponding to the set $\sigma_*(k) \subset \R$. Assume that
$\sigma_{*}$ satisfies (\ref{Gap condition}). Then the family $\{
P_*(k) \}_{k \in \R^d}$ has the following properties:
\begin{enumerate}
    \item[$\mathrm{(P_1)}$] the map $k \mapsto P_*(k)$ is smooth from $\R^d$ to
    $\B(\Hf)$ (equipped with the operator norm);
    \item[$\mathrm{(\tilde{P}_{1})}$] the map $k \mapsto P_*(k)$
    extends to a $\B(\Hf)$-valued analytic function on the domain
    \begin{equation}\label{Omega}
        \Omega_{\alpha} = \left\{ \kappa \in \C^d: \ |\mathrm{Im}(\kappa_j)| <
        \alpha \quad \forall j \in \set{1, \ldots, d} \right\}
    \end{equation}
    for some $\alpha >0$;
    \item[$\mathrm{(P_2)}$] the map $k \mapsto P_*(k)$ is
    $\tau$-covariant, \ie
    \[
    P_*(k + \la) = \ta(\la)^{-1} \, P_*(k) \, \ta(\la)  \qquad \forall k
    \in \R^d, \quad \forall \la \in \Gamma^{*};
    \]
    \item[$\mathrm{(\tilde{P}_2)}$] the map $\kappa \mapsto P_*(\kappa)$ is
    $\tau$-covariant, \ie
    $$
    P_*(\kappa + \la) = \ta(\la)^{-1} \, P_*(\kappa) \, \ta(\la)  \qquad \forall
    \kappa \in \Omega_{\alpha}, \quad \forall \la \in \Gamma^{*};
    $$
    \item[$\mathrm{(P_3)}$] there exists an antiunitary
    operator\footnote{\ By \emph{antiunitary} operator we mean a surjective antilinear operator $C: \Hi \rightarrow \Hi$,
    such that
    $\< C\ph, C \psi \>_{\Hi} = \< \psi, \ph \>_{\Hi}$ for any $\ph, \psi \in \Hi$.} 
     $C$ acting on $\Hf$ such that
    \[
    P_*(-k) =  C \, P_*(k) \, C^{-1}  \qquad \mbox{ and  } \qquad C^{2} =1.
    \]
\end{enumerate}
\end{proposition}

While properties $\mathrm{(P_1)}$ and $\mathrm{(P_2)}$ are a
consequence of the fact that the operator $H$ commutes with the
lattice translations, jointly with the gap condition (\ref{Gap
condition}), property $\mathrm{(P_3)}$ follows from the fact that
the operator (\ref{Hamiltonian}) is real, and corresponds to the
time-reversal symmetry of the physical system. Notice that property
$\mathrm{(P_3)}$ generically does not hold true for periodic
\emph{magnetic} Schr\"odinger operators.

\begin{proof}
One uses (\ref{H(k)}) to define $H(\kappa)$ for every $\kappa \in
\C^d$. Since for any $\kappa_0 \in \C^d$ one has
$$
H(\kappa) = H(\kappa_0) + 2 (\kappa - \kappa_0) \cdot (-i \nabla_y)
+ (\kappa^2 - \kappa_0^2) \1
$$
and $(-i \nabla_y)$ is $H(\kappa_0)$-bounded, the family
$\{H(\kappa)\}_{\kappa \in \C^d}$ is an entire analytic family of
type (A) \cite[Chapter XII]{RS4}. Hence $\{H(\kappa)\}_{\kappa \in
\C^d}$ is an entire analytic family in the sense of Kato
\cite[Theorem XII.9]{RS4}. As a consequence, the set
$$
R = \set{(\kappa, \lambda)\in \C^d \times \C : \ \lambda \in
\rho(H(\kappa))}
$$
is an open set, and $(\kappa, \lambda) \mapsto (H(\kappa) - \lambda
\1)^{-1}$ is an analytic function on $R$. Since $R$ is open and
(\ref{Gap condition}) holds, for every every $k_0 \in \R^d$ there
exist a positively oriented circle $\Lambda(k_0) \subset \C $ and a
neighborhood $N_{k_0} \subset \C^d$ of $k_0$, such that
$\Lambda(k_0)$ separates $\sigma_{*}(\kappa)$ from the rest of the
spectrum and $\Lambda(k_0) \subset \rho(H(\kappa))$ for every
$\kappa \in N_{k_0}$. Then the Riesz's formula
\begin{equation}\label{Riesz formula}
P_*(\kappa) = {\frac{i}{2 \pi}} \oint_{\Lambda(k_0)} (H(\kappa) -
z\1)^{-1} dz
\end{equation}
defines a map $\kappa \mapsto P_*(\kappa)$ which is analytic in
$N_{k_0}$. For $\kappa = k \in N_{k_0} \cap \R^d$, it agrees with
the spectral projection $P_*(k)$, so $\mathrm{(P_1)}$ holds true.

Since $\sigma_*$ is $\Gamma^*$-periodic, the circle in (\ref{Riesz
formula}) can be chosen so that $\Lambda(k_0) = \Lambda(k_0 + \la)$
for every $\la \in \Gamma^*$. Thus, property $\mathrm{({P}_2)}$
follows from (\ref{H covariant}) and (\ref{Riesz formula}).

Formula (\ref{Riesz formula}) yields a $\B(\Hf)$-valued analytic
function on $\cup_{k_0} N_{k_0} \supset \R^d$. In view of
$\mathrm{(P_2)}$ and the unique continuation principle, one may
assume $N_{k_0 + \lambda} = N_{k_0}$ for every $\lambda \in
\Gamma^*$, so by the compactness of $Y^*$ there exist $\alpha > 0$
such that $\cup_{k_0} N_{k_0} \supset \Omega_{\alpha}$, \ie
$\mathrm{(\tilde{P}_1)}$ holds true. Moreover, in view of
$\mathrm{(P_2)}$ and the unique continuation principle, one gets
$\mathrm{(\tilde{P}_2)}$.

%


Property $\mathrm{(P_3)}$ corresponds to time-reversal symmetry.
This symmetry  is realized in $L^2(\R^d)$ by the complex conjugation
operator, acting as $( T \psi )(x) = \overline{\psi(x)}$ for $\psi
\in L^2(\R^d).$ One directly checks that $\tilde{T}= \UZ \, T \,
\UZ^{-1}$ acts on $\Ht \simeq L^2(Y^*, \Hf)$ as
\[
( \tilde{T} \ph )(k) = C \, \ph(-k), \qquad  \qquad \ph \in
L^2(Y^{*},\Hf),
\]
where $C$ is the complex conjugation operator in $\Hf$. From the
fact that $H$ commutes with $T$, taking into account (\ref{Gap
condition}), property $\mathrm{(P_3)}$ follows, see \cite{Pa2008}
for details.
\end{proof}



\medskip

\subsection{The Wannier functions and their localization properties}$ $\\

\noindent \textbf{Case I. Simple Bloch band.} We initially consider
the case of a single isolated Bloch band, namely $\sigma_*(k) =
\set{E_n(k)}$ and (\ref{Gap condition}) is satisfied, such that
$E_n(k)$ is an eigenvalue of multiplicity one for every $k$. This
case corresponds to the \virg{simple direct isolated band} in
\cite{Ne91}.

A (normalized) \textbf{Bloch function} corresponding to the $n$-th
Bloch band is, by definition, any  $\ph \in \Ht$ satisfying
\begin{equation}\label{Eigen equation H}
\ph(k, \cdot) \in \Do_0  \qquad H(k)\ph(k, \cdot) = E_n(k) \ph(k, \cdot) \quad \mbox{ and }\quad
\norm{\ph(k, \cdot)}_{\Hf} =1        \qquad \forall k \in Y^*.
\end{equation}

\noindent Clearly, if $\ph$ is a Bloch function then $\tilde{\ph}$,
defined by $\tilde{\ph}(k,y)= \expo{\I \vartheta(k)} \ph(k,y)$ for
any function $\vartheta: \Base \to \R$, is also a Bloch function.
The latter invariance is often called \emph{Bloch gauge invariance}.
Notice that (\ref{Eigen equation H}) is equivalent to
\begin{equation}\label{Eigen equation P1}
P_{*}(k) \ph(k, \cdot) = \ph(k, \cdot) \quad \mbox{ and }\quad
\norm{\ph(k, \cdot)}_{\Hf} =1        \qquad \forall k \in Y^*,
\end{equation}
where $P_{*}(k)$ is the rank-one projection on the eigenspace
corresponding to the eigenvalue $E_n(k)$.

\begin{definition}\label{Def Wannier}
The \textbf{Wannier function} $w_n \in L^2(\R^d)$ corresponding to a
Bloch function $\ph_n \in \Ht$ for the Bloch band $E_n$ is the
preimage of $\ph_n$ with respect to the Bloch-Floquet transform,
namely
$$
w_n(x) := \( \UZ^{-1} \ph_n \)(x) = \frac{1}{|Y^*|^{\half}} \int_{Y^*} dk \  \expo{\I k
\cdot x} \ph_n(k, [x]).
$$
\end{definition}

The translated Wannier functions are
      $$
      w_{n,\gamma}(x) := w_n(x - \gamma) =  \frac{1}{|Y^*|^{\half}} \int_{Y^*} dk \ \expo{- \I k \cdot \gamma} \ \expo{ \I k \cdot
x} \ph_n(k, [x]),   \qquad \gamma \in \Gamma.
      $$
Thus, in view of the orthogonality of the trigonometric polynomials
and the fact that $\UZ$ is an isometry, the functions $\{ w_{n,
\gamma} \}_{\gamma \in \Gamma}$ are mutually orthogonal in
$L^2(\R^d)$. Moreover, the family $\{ w_{n, \gamma}\}_{\gamma \in
\Gamma}$ is a complete orthonormal basis of $\UZ^{-1} \, \(\ran
P_{*}\)$, where $P_{*} = \int^{\oplus}_{Y^*} P_{*}(k) \, dk$ is the
total projector corresponding to the Bloch band $E_n$. Notice that
the previous definition and elementary properties do not rely on the
gap condition (\ref{Gap condition}).

The (weak) localization of $w_n$ expressed by the fact that $w_n$ is
in $L^2(\R^d)$ is physically unsatisfactory, since it does not imply
that the position operator and its powers have finite expectation
value. Therefore, it is natural to search for conditions on $\ph_n$
which guarantee a stronger decay of $w_n$. In a nutshell, the
localization properties of the Wannier function are determined by
the regularity (continuity, smoothness, analyticity, \ldots) of the
corresponding Bloch function. By a simple integration by parts, and
an exact cancelation of the opposite boundary terms, one obtains the
following lemma. Here $X_j$, for $j \indexd$, denotes the the $j^{\,
\rm th}$ component of the position operator, \ie $(X_j \psi)(x) =
x_j \, \psi(x)$ for $\psi$ in a suitable dense subspace of
$L^2(\R^d)$

\begin{lemma}\label{Lemma derivative}
Let $\ph \in \Ht$, $\ph \in C^1(\R^d, \Hf)$ and let $w = \UZ^{-1}
\ph$. Then $X_j w $  is in $L^2(\R^d)$ and
$$
\UZ (-\I X_j w) = \partial_{k_j} \ph.
$$
\end{lemma}

\noindent By iterating the previous lemma, and taking into account
(\ref{Zak equivalence}), one concludes that if $\ph \in \Ht$ is in
$C^{\infty}(\R^d, \Hf)$, then the corresponding Wannier function
decreases faster than the inverse of any polynomial, \ie $P(X_1,
\ldots, X_d) w$ is in $ L^2(\R^d)$ for any polynomial $P$. As for
the exponential localization, by mimicking the proof of the usual
Paley-Wiener theorem one gets the following result, see also
\cite{Kuchment_book} for a slightly different formulation.

\begin{proposition}\label{Lemma Paley-Wiener}
Let $\Omega_{\beta}$ be defined by (\ref{Omega}) and
$$
\Hi_{\ta, \, \beta}^{\C} = \set {\Phi \in L^2_{\rm
loc}(\Omega_{\beta}, \Hf): \Phi(z - \la) = \ta(\la) \Phi(z) \mbox{
for all } \la \in \Gamma^*, z \in \Omega_{\beta}}.
$$
Let $\phi$ be the restriction to $\R^d$ of a function $\Phi \in
\Hi_{\ta, \, \beta}^{\C}$ analytic in the strip $\Omega_{\rm
\beta}$. Assume that
$$
\int_{Y^*} \norm{\Phi( k + \I h)}_{\Hf}^2 \, dk \leq C \qquad
\forall h \mbox{ with } |h_j| < \beta
$$
with a constant $C$ uniform in $h$. Then, the function $w :=
\UZ^{-1} \, \phi$ satisfies
$$  
\int_{\R^d} \expo{2 \beta |x|} \, |w(x)|^2 dx <
+ \infty. $$  
\end{proposition}

\begin{proof}[Sketch of the proof] Let $k + \I h$ be in
$\Omega_{\beta}$, and pose $\phi^h(k):= \Phi(k + \I h)$, so that
$\phi_h \in \Ht$. By shifting the integration contour by the method
of residues, and by the $\tau$-equivariance of $\phi^h$, one gets
$$
w^h(x) := \( \UZ^{-1} \phi^h\)(x) = \int_{Y^*} dk \, \expo{\I k
\cdot x} \, \Phi(k + \I h, [x])  = \expo{h \cdot x} w(x).
$$
Then, by the unitarity of $\UZ$,
$$
\norm{\expo{h \cdot x} w}^2_{L^2(\R^d)} = \int_{Y^*} \norm{\Phi(k + \I
h)}^2_{\Hf} \, dk \leq C.
$$
Since the latter constant does not depend on $h$ for $|h_j| \leq
\beta$, one has $\norm{\expo{\beta |x|} w}_{L^2(\R^d)} < + \infty$.
\end{proof}

\begin{corollary}\label{Cor Exp local}
Let $\phi$ be the restriction to $\R^d$ of a function $\Phi \in
\Hi_{\ta, \, \alpha}^{\C}$ analytic in the strip $\Omega_{\rm
\alpha}$. Then, for every $\beta < \alpha$, the function $w :=
\UZ^{-1} \, \phi$ satisfies
\begin{equation}\label{Exp localization beta}
\int_{\R^d} \expo{2 \beta |x|} \, |w(x)|^2 dx < + \infty.
\end{equation}
\end{corollary}


\noindent A function $w \in L^2(\R^d)$ satisfying (\ref{Exp
localization beta}) for some $\beta >0$ is said to be
\textbf{\emph{exponentially localized}}, while a function $w \in
L^2(\R^d)$  such that $P(X_1, \ldots, X_d) w$ is in $ L^2(\R^d)$ for
any polynomial $P$ is said \textbf{\emph{almost-exponentially
localized}}.


In view of the previous proposition, the existence of an
exponentially (resp.\ almost-expo\-nentially) localized Wannier
function for the Bloch band $E_n$ is equivalent to the existence of
an analytic (resp.\ smooth) Bloch function. Property
$\mathrm{(\tilde{P}_1)}$ (resp. $\mathrm{(P_1)}$) assures that there
is a choice of the Bloch gauge such that the Bloch function is
analytic (resp. smooth) around a given point. However, as several
authors noticed \cite{Cl1, Ne91}, there might be a topological
obstruction to obtaining a global analytic (resp.\ smooth) Bloch
function, in view of the competition between the regularity and the
$\tau$-equivariance (remember that the Bloch function must be in
$\Ht$ by definition). This topological obstruction will be encoded
in the concept of Bloch bundle, which we will introduce in the next
subsection. Preliminarily, we describe the more realistic case of
the composite Bloch bands.

\medskip

\noindent \textbf{Case II. Composite Bloch bands.} We consider now
the generic case of a family $\sigma_{*}$ of $m$ Bloch bands
satisfying (\ref{Gap condition}). Since the eigenvalues in the
family $\sigma_{*}$ generically intersect each other, and the
eigenprojectors of $H(k)$ corresponding to single eigenvalues are
not smooth at the intersection point, generically
it is not even possible to find a system of locally smooth Bloch
functions spanning $\Ran P_{*}(k)$ at any $k$ close to a given
$k_0$. Thus, the notion of Bloch function is relaxed and replaced by
the following one \cite{Bl, Cl1}.

\begin{definition}\label{Def quasiBloch} Let $\{ P_{*}(k) \}_{k \in \R^d}
\subset \B(\Hf)$ be a family of orthogonal projectors satisfying
$\mathrm{(P_1)}$ and $\mathrm{(P_2)}$, with $\dim P_{*}(k) \equiv m
< + \infty$. A function $\chi \in \Ht$ is called a
\textbf{quasi-Bloch function} {\upshape(}for the family $\{ P_{*}(k)
\}${\upshape)}  if
\begin{equation}\label{P*condition}
P_{*}(k) \chi(k, \cdot) =  \chi(k, \cdot) \quad \mbox{ and } \quad
\chi(k, \cdot) \neq 0  \qquad \forall k \in Y^*.
\end{equation} A \textbf{Bloch frame} {\upshape(}for the family $\{ P_{*}(k)
\}${\upshape)} is a set $\{ \chi_a \}_{a = 1, \ldots, m}$ of
quasi-Bloch functions such that $\{\chi_1(k), \ldots, \chi_m(k)\}$
is an orthonormal basis of\ $\ran P_{*}(k)$ at (almost-)every $k \in
Y^*$.
\end{definition}

A Bloch frame is fixed only up to a $k$-dependent unitary matrix
$U(k) \in \Un(m)$, \ie if $\{ \chi_a \}_{a = 1, \ldots, m}$ is a
Bloch frame then the functions $
  \widetilde \chi_a (k)  =  \sum_{b =1}^{m}  \chi_b(k) U_{b,a}(k)$
also define a Bloch frame.

\begin{definition}\label{Composite Wannier function}
The \textbf{composite Wannier functions} corresponding to a Bloch
frame $\{ \chi_a \}_{a = 1}^{m}$ are the functions
$$
w_a(x) := \( \UZ^{-1} \chi_a\)(x), \qquad \qquad a \in \set{1,
\ldots, m}.
$$
\end{definition}

As in the case of a single Bloch band, the existence of a system of
exponentially localized (resp.\ almost-exponentially localized)
composite Wannier function is equivalent to the existence of an
analytic (resp.\ smooth) Bloch frame. The topological obstruction to
the existence of a regular (analytic, smooth or continuous) Bloch
frame, already observed in \cite{Cl1, Ne91} is described in the next
subsection.

\begin{ciao}[\textbf{Regularity of composite Wannier functions}] 
\label{Rem regularity Wannier} We emphasize that, for $V_{\Gamma}$
satisfying (\ref{V assumptions}), the composite Wannier functions
are actually in $W^{2,2}(\R^d)$ and, in general, one cannot expect
better regularity properties. For example, if $V_{\Gamma}$ has a
Coulomb singularity, the Wannier functions are not smooth, as it
happens for the eigenfunctions of the hydrogen atom.

To study the regularity properties of  $w_a$,  one notices that the
corresponding quasi-Bloch function $\psi_a(k)$ is in
$W^{2,2}(\T^d)$ and $\|\psi_a(k)\|_{W^{2,2}(\T^d)} \leq  C_m \,
\norm{\psi_a (k)}_{L^2(\T^d)}$. Indeed, one can check that the
previous inequality is a consequence of the fact that $V_{\Gamma}$
corresponds to a multiplication operator in $L^2(\T^d)$ which is
$\Lap$-bounded with relative bound zero, together with the fact that
the union of the ranges of the functions $\{ E_i \, : \,n \leq i
\leq n+m-1 \}$ is contained in a fixed compact set.   As a
consequence, in view of (\ref{Zak equivalence}), one concludes that
$w_a \in W^{2,2}(\R^d)$, namely $w_a$ is in the domain of $H$.
\end{ciao}

\begin{ciao}
[\textbf{Comparison with the physics literature}]\label{Rem physical literature}
When comparing our definitions with the physics literature, one has
to take into account that we are using a \emph{modified}
Bloch-Floquet transform, so that the fiber Hamiltonian (\ref{H(k)}) has a $k$-independent
domain. This fact is convenient from the mathematical viewpoint. Alternatively, the classical
Bloch-Floquet transform
\begin{equation}\label{BF transform}
    (\UB \psi)(k,y):= \frac{1}{|Y^*|^{\half}}
\sum_{\gamma\in\Gamma} \expo{- \I k \cdot \gamma } \,
\psi(y+\gamma), \qquad k \in \R^d, \quad y\in\R^d,
\end{equation}
yields a decomposition $\UB \, \(- \Lap + V_{\Gamma}\) \, \UB^{-1} =
\int^{\oplus}_{\Base} H_{\rm cl}(k) \, dk$ where $H_{\rm cl}(k)= -
\Lap + V_{\Gamma}$ acts in
$$
\Hi_{k} := \set{ \psi \in L^2_{\rm loc}(\R^d): \psi(y + \gamma)=
\expo{\I k \cdot \gamma} \psi(y) \quad \forall \gamma \in \Gamma
\quad \mbox{for a.e. }y \in\R^d}
$$
on the domain  $\Do_{k} := W^{2,2}_{\rm loc}(\R^d) \cap \Hi_k$. The
unitary operator $\mathcal{J} = \UB \, \UZ^{-1}$ acts as $ \(\J \ph
\)(k,y) = \expo{\I k \cdot y} \ph(k,y)$, is a fibered operator and its
fiber, denoted by $J(k)$, maps unitarily the space $\Hf$ into the
space $\Hi_k$. Since $J(k) H_{\rm cl}(k) J(k)^{-1} = H(k)$, the
operators $H(k)$ and $H_{\rm cl}(k)$ have the same spectrum $\{
E_n(k) \}_{n \in \N}$.

In the physics literature, usually a Bloch function is defined as an
eigenfunction $\psi_n(k, \cdot) \in \Do_k$ of $H_{\rm cl}(k)= - \Lap
+ V_{\Gamma}$ for the eigenvalue $E_n(k)$. By the unitary
equivalence above, the so-called \emph{Bloch theorem}, one has that
$\psi_n(k, y) = \expo{\I k \cdot y} \ph_n(k, y)$, where $\ph_n(k,
\cdot)$ is $\Gamma$-periodic and $\ph_n \in \Ht$. The latter is our
Bloch function, as defined by (\ref{Eigen equation H}). Consequently,
in the physics literature the Wannier function is defined as $w_n =
\UB^{-1} \psi_n = \UZ^{-1} \ph_n$, which coincides exactly with our
definition.
\end{ciao}


\goodbreak

\subsection{The Bloch bundle} To describe the geometric obstruction
to the existence of a continuous, smooth or analytic Bloch frame it
is convenient to introduce the concept of Bloch bundle, following
\cite{Pa2008}. In this subsection we assume as given a family of
orthogonal projectors $\{ P_*(k) \}_{k \in \R^d}$  satisfying
properties $\mathrm{({P}_1)}$, $\mathrm{(P_2)}$ and
$\mathrm{(P_3)}$, or their complex counterparts. Notice that we
could also abstract from the specific case of the operator
(\ref{Hamiltonian}) and consider such properties as convenient
starting assumptions. Within this viewpoint, our approach can be
applied to the periodic Pauli or Dirac operator, with obvious
modifications.

\begin{proposition}\label{Prop Bloch bundle}
To a family of orthogonal projectors $\{ P_*(k) \}_{k \in \R^d}$
satisfying $\mathrm{(P_1)}$ and $\mathrm{(P_2)}$ is canonically
associated a Hermitian smooth vector bundle $\E_* $ over $\T_d^*$,
called the \emph{Bloch bundle}. If $\mathrm{(\tilde{P}_1)}$ and
$\mathrm{(\tilde{P}_2)}$ are also satisfied, then $\E_*$ is the
restriction to\  $\Base= \R^d/{\Gamma^*}$ of a holomorphic Hermitian
vector bundle $\tilde{\E}_*$ over $\Omega_{\alpha}/{\Gamma^*}$.
\end{proposition}

\begin{proof}
The idea of the construction is to firstly consider  $\sqcup_{k \in \R^d} \Ran P_*(k)$ as
a subbundle of the trivial bundle  $\R^d \times \Hf$ over $\R^d$. Then
we use the $\tau$-equivariance of the projectors to obtain a (quotient)
vector bundle $\E_*$ over the quotient space $\T_d^* = \R^d /\Gamma^*$.

To construct $\E_*$, one firstly introduces on the set $\R^d \times
\Hf$ the equivalence relation $\sim_{\ta}$, where
\[
(k,\ph) \sim_{\ta} (k', \ph')  \qquad  \Leftrightarrow  \qquad (k',
\ph')= (k + \la \,, \, \ta(\la)^{-1} \ph) \quad \mbox{for some } \la
\in \Gamma^*.
\]
The equivalence class with representative $(k, \ph)$ is denoted by
$[k,\ph]$. Then the total space $\E_*$ of the vector bundle is
defined by
\[
\E_* := \left \{ [k, \ph] \in (\R^d \times \Hf )/{\sim_{\ta}} :
\quad \ph \in \Ran P_*(k) \right \}.
\]
This definition does not depend on the representative in view of the
covariance property $\mathrm{(P_2)}$. The projection to the base
space $\pi: \E_* \rightarrow \T_d^*$ is $\pi[k, \ph] = \mu(k)$,
where $\mu$ is the projection modulo $\Gamma^*$. One checks that $
\E_* \stackrel{\pi}{\rightarrow} \Base$ is a smooth complex vector
bundle with typical fiber $\C^m$. In particular, the local
triviality follows from $\mathrm{(P_1)}$ and the use of the
Kato-Nagy formula. Indeed, for any $k_0 \in \R^d$ there exists a
neighborhood $O_{k_0} \subset \R^d$ of $k_0$ such that $\| P_*(k) -
P_*(k_0) \| < 1$ for any $k \in O_{k_0}$. Then by setting
(Kato-Nagy's formula \cite[Sec. I.6.8]{Kato})
\begin{equation}\label{W definition}
W(k) := \( 1 - (P_*(k) - P_*(k_0))^2 \)^{-1/2} \big( P_*(k)P_*(k_0)
+ (1 - P_*(k))(1 - P_*(k_0))\big)
\end{equation}
one gets a smooth map $W: O_{k_0} \rightarrow \Un(\Hf)$ such that
$W(k) \, P_*(k_0) \, W(k)^{-1} = P_*(k)$. If $\{ \chi_{a}
\}_{a=1,\dots,m}$ is any orthonormal basis spanning $\mathrm
{Ran}P_*(k_0)$, then $\ph_a(k) = W(k) \chi_a$ is a smooth local
orthonormal frame for $\E_*$, yielding the local triviality of the
fibration $\E_* \stackrel{\pi}{\rightarrow} \Base$.

The vector bundle $\E_*$ carries a natural Hermitian structure.
Indeed, if $v_1, v_2 \in \E_*$ are elements of the fiber over $\mu(k) \in
\Base$ then, up to a choice of the representatives, one has $v_1 =
[k, \ph_1]$ and $v_2 = [k, \ph_2]$, and one poses $ \< v_1, v_2
\>_{{\E_{*}}} := \< \ph_1, \ph_2 \>_{\Hf}. $

As for the analytic case, one introduces an equivalence relation
$\sim_{\ta}$ over $\Omega_{\alpha} \times \Hf$ as above. Then the
total space of $\tilde{\E}_*$ is defined by
$$
\tilde{\E}_* = \left\{ [\kappa, \ph] \in (\Omega_{\alpha} \times \Hf
)/{\sim_{\ta}} :\ \quad \ph \in \Ran P_*(\kappa) \right\}
$$
which is independent of the representative in view of
$\mathrm{(\tilde{P}_2)}$. Local triviality follows again from
formula (\ref{W definition}). Indeed, for every $\kappa_0 \in
\Omega_{\alpha}$ there exists a neighborhood $N_{\kappa_0} \subset
\Omega_{\alpha}$ such that $\| P_*(\kappa) - P_*(\kappa_0) \| < 1$
for every $\kappa \in N_{\kappa_0}$; then, formula (\ref{W
definition}) yields a holomorphic map $W: N_{\kappa_0} \to \Un(\Hf)$
such that $W(\kappa) P_*(\kappa_0)W(\kappa)^{-1} = P_*(\kappa)$.
Notice that, since $P_*(\kappa)^* = P_*(\overline{\kappa})$, the
operator $W(\kappa)$ is not unitary when $\kappa \notin \R^d \cap
N_{\kappa_0}$.

Since $\kappa \mapsto P_*(\kappa)$ extends $k \mapsto P_*(k)$ by
$\mathrm{(\tilde{P}_1)}$, and both the maps are $\tau$-covariant,
${\E}_* \rightarrow \Base$ is clearly a restriction of $\tilde{\E}_*
\rightarrow \Omega_{\alpha}/\Gamma^*$. The Hermitian structure over
the vector bundle $\tilde{\E}_*$ is defined as in the smooth case.
\end{proof}

The vector bundle $\E_*$  is equipped with a natural
$\mathfrak{u}(m)$-connection (\emph{Berry connection}), induced by
the trivial connection on the trivial vector bundle $(\R^d \times
\Hf )/{\sim_{\ta}} \to \Base$. The triviality of the latter is a
consequence of the fact that $\Un(\Hf)$ is contractible whenever
$\Hf$ is infinite-dimensional \cite{Kui}.

The Bloch bundle encodes the geometrical obstruction to the
existence of a global smooth (resp. analytic) Bloch frame. Indeed,
the following result is implicit in \cite{Pa2008}.

\begin{theorem}\label{Th Bloch bundle}
Let $\{ P_*(k) \}_{k \in \R^d} \subset \B(\Hf)$ be a family of
orthogonal projectors satisfying properties $\mathrm{({P}_1)}$
$($resp.\ $\mathrm{(\tilde{P}_1)}$$)$ and $\mathrm{(P_2)}$ $($resp.\
$\mathrm{(\tilde{P}_2)}$$)$, with $\dim P_*(k) \equiv m < + \infty$.
Then the following statements are equivalent:
\renewcommand{\labelenumi}{{\rm(\Alph{enumi})}}
\begin{enumerate}
    \item[$\mathrm(A)$]\textbf{existence of a regular Bloch frame:} there exists
    a Bloch frame $\set{\chi_a}_{a=1, \ldots, m}$ such that each $\chi_a$ is in
    $C^{\infty}(\R^d,\Hf)$ $($\,resp.\ each $\chi_a$ is the restriction to $\R^d$ of a
    function $\tilde{\chi}_a \in \Hi^{\C}_{\tau, \alpha}$ analytic on
    $\Omega_{\alpha}$\,$)$;\vspace{3mm}
    \item[$\mathrm{(B)}$]\textbf{triviality of the Bloch bundle:}
    the vector bundle  $\E_*$, associated to the family $\{ P_*(k) \}_{k \in \R^d}$ according to
    Proposition \ref{Prop Bloch bundle}, is trivial in the category of smooth Hermitian
    vector bundles over $\Base$ $($\,resp.\ the vector bundle $\tilde{\E}_*$ is trivial in the
    category of holomorphic Hermitian vector bundles over $\Omega_{\alpha}/{\Gamma^*}$$)$.
\end{enumerate}
\end{theorem}

The Bloch bundle can be trivial for reasons unrelated to
time-reversal symmetry, as in some phases of the Haldane model \cite{Haldane88,
Prodan11}. On the other hand, as a consequence of \cite{Pa2008}, under the assumption of
time-reversal symmetry the Bloch bundle is always trivial in low dimension, as
stated in the following result.

\begin{theorem}\label{Theorem triviality}
Let $\{ P_*(k) \}_{k \in \R^d} \subset \B(\Hf)$ be a family of
orthogonal projectors satisfying properties
$\mathrm{(\tilde{P}_1)}$, $\mathrm{(\tilde{P}_2)}$ and
$\mathrm{(P_3)}$. Assume $d \leq 3$ and $m \geq 1$, or $d \geq 1$
and $m=1$. Then there exists a Bloch frame  $\set{\chi_a}_{a=1,
\ldots, m}$ such that each $\chi_a$ is real-analytic, \ie   $\chi_a
\in C^{\omega}(\R^d,\Hf) \cap \Ht$.
\end{theorem}

\begin{proof}[Sketch of the proof]
Since the family $\{ P_{*}(k)\}_{k \in \R^d}$ satisfies properties
$\mathrm{(P_1)}$, $\mathrm{(P_2)}$ and $\mathrm{(P_3)}$, in view of
\cite[Theorem 1]{Pa2008} there exists a smooth Bloch frame $\chi =
\set{\chi_1, \ldots, \chi_m} \subset \Ht \cap C^{\infty}(\R^d,
\Hf)$, \ie the Bloch bundle is trivial as a smooth Hermitian vector
bundle. Moreover, since $\mathrm{(\tilde P_1)}$ is also satisfied,
$\{ P_{*}(k)\}_{k \in \R^d}$ admits a holomorphic extension to
$\Omega_{\alpha}$, see (\ref{Omega}). Theorem 2 in \cite{Pa2008}
implies the existence of a family of holomorphic functions
$\chi_a^{\C}: \Omega_{\alpha} \to \Hf$, $\chi_a \in \Hi_{\ta, \,
\alpha}^{\C}$, such that $\set{\chi^{\C}_1(\kappa), \ldots,
\chi^{\C}_m(\kappa)}$ is a (possibly non-orthonormal) basis of $\Ran
P_*(\kappa)$ for every $\kappa \in \Omega_{\alpha}$. By restriction
from $\Omega_{\alpha}$ to $\R^d$, and by a Gram-Schmidt
orthonormalization procedure, one gets a real-analytic Bloch frame $
\{\chi_1, \ldots, \chi_m \} \subset C^{\omega}(\R^d,\Hf) \cap \Ht$.
\end{proof}


\section{The Marzari-Vanderbilt localization functional}

The long-lasting uncertainty about the existence of exponentially
localized composite Wannier functions in three dimensions, settled
only recently \cite{Wannier_letter_BPCM, Pa2008}, and the need of an
approach suitable for numerical simulations, lead the solid state
physics community to explore new paths. In an important paper
\cite{MaVa}, Marzari and Vanderbilt introduced the following
concept.

For a single-band normalized Wannier function $w\in
L^2(\mathbb{R}^d)$, one defines the localization functional by
\begin{equation}
\label{mvw1banda} F_{MV}(w)=\sum_{j=1}^d \Var \left(X_j; |w(x)|^2dx
\right) = \int_{\mathbb{R}^d} |x|^2 |w(x)|^2dx - \sum_{j=1}^d \left(
\int_{\mathbb{R}^d} x_j |w(x)|^2dx \right)^2,
\end{equation}
which is well-defined at least whenever $\int_{\mathbb{R}^d} |x|^2
|w(x)|^2dx < + \infty$.

More generally, for a system of $L^2$-normalized composite Wannier
functions $w= \{ w_1, \ldots, w_m \} \subset L^2(\R^d)$ the
\textbf{Marzari-Vanderbilt localization functional} is
\begin{equation}
\label{mvwmbande} F_{MV}(w)=\sum_{a=1}^m F_{MV}(w_a) = {\sum_{a=1}^m
} \int_{\mathbb{R}^d} |x|^2 |w_a(x)|^2dx - \sum_{a=1}^m \sum_{j=1}^d
\left( \int_{\mathbb{R}^d} x_j |w_a(x)|^2dx \right)^2.
\end{equation}
We emphasize that the above definition of $F_{\rm MV}(w)$ includes the crucial
constraint that the corresponding Bloch functions $\varphi_a(k,
\cdot)=( \UZ\, w_a ) (k, \cdot)$, for $a \in \{ 1, \ldots,m \}$, are
a Bloch frame, \ie $\{ \varphi_1(k, \cdot), \ldots, \varphi_m(k,
\cdot)\}$ is an orthonormal set in $\mathcal{H}_f$ for each $k \in
Y^*$ and
\begin{equation}
\label{constraint} \Span_{\mathbb{C}}\left\{\varphi_1(k, \cdot),
\ldots, \varphi_m(k, \cdot)\right\}=P_*(k)(\mathcal{H}_f) \, , \quad
\forall k \in Y^* \,.
\end{equation}
According to Remark \ref{Rem regularity Wannier}, the latter
condition actually implies $w_a \in W^{2,2}(\R^d) = \Do(H)$.

\begin{definition}\label{Def Optimally loc Wannier functions} Let
$\{ P_*(k) \}_{k \in \R^d} \subset \B(\Hf)$ be a family of
projectors satisfying properties $\mathrm{(\tilde{P}_1)}$ and
$\mathrm{(P_2)}$, with $\dim P_*(k) \equiv m < + \infty$. A system
of \textbf{maximally localized composite Wannier functions} is a
global minimizer $\{w_1, \ldots, w_m \} $ of the Marzari-Vanderbilt
localization functional $F_{\rm MV}$ in the set $\W^m := \(\Do(H)
\cap \Do(X)\)^m$, under the constraint that $\{\ph_1, \ldots, \ph_m
\}$, for $\ph_a = \UZ w_a$, is a Bloch frame.
\end{definition}

A natural problem, raised in \cite{MaVa}, is the following.
\begin{problem}\label{Problem MV} Let
$\{ P_*(k) \}_{k \in \R^d} \subset \B(\Hf)$ be a family of
projectors satisfying properties $\mathrm{(\tilde{P}_1)}$ and
$\mathrm{(P_2)}$, with $\dim P_*(k) \equiv m < + \infty$.
\begin{enumerate}
    \item[$\mathrm{(MV_1)}$] \textbf{(Existence)} prove that there exists a system of
    maximally localized composite Wannier functions;
    \item[$\mathrm{(MV_2)}$]\textbf{(Localization)} prove that any maximally localized composite
    Wannier function is exponentially localized, in the sense that there exists $\beta >0$ such
    that (\ref{Exp localization beta}) holds.
\end{enumerate}
\end{problem}
Since the (modified) Bloch-Floquet transform $ \UZ:
L^2(\mathbb{R}^d) \to L^2(Y^*; \Hf)$ is an isometry and it satisfies
$(\tilde{\mathcal{U}}_{BF} {X_j g})(k,y)= \I
\frac{\partial}{\partial k_j} (\tilde{\mathcal{U}}_{BF} g)(k,y)$,
the functional \eqref{mvwmbande} can be rewritten in terms of the
Bloch frame $\varphi=\{\varphi_1,\ldots, \varphi_m\} $ as
\begin{equation}
\label{mvfimbande} \tilde{F}_{MV}(\varphi)=\sum_{a=1}^m\sum_{j=1}^d
\left\{ \int_{Y^*} dk \int_{\T^d} \left| \frac{\partial
\varphi_a}{\partial k_j}(k,y) \right|^2dy - \left( \int_{Y^*} dk
\int_{\T^d}  \overline{\varphi_a(k,y)} \, i\frac{\partial
\varphi_a}{\partial k_j}(k,y)\, dy \right)^2  \right\}\, .
\end{equation}
Correspondingly, in view of (\ref{Zak equivalence}), the space $\W =
\Do(H) \cap \Do(X) = \Do(H) \cap \Do(\<X\>)$ is mapped by the
Bloch-Floquet transform into
$$
\Ht \cap L^2_{\rm loc}(\R^d, W^{2,2}(\T^d)) \cap W^{1,2}_{\rm
loc}(\R^d, L^2(\T^d)) =: \widetilde{\W}.
$$

Hereafter, to solve problem $\mathrm{(MV_1)}$ for any $d$,  we will make the following

\noindent \textbf{Assumption 1:} there exists a Bloch frame $\chi=\{
\chi_1,\ldots, \chi_m\} \subset \Ht$ such that $\chi_a \in
\widetilde{\W}$.

\noindent In the case most relevant to us, \ie Schr\"odinger operators
for $d \leq 3$, the previous assumption is automatically satisfied, since Proposition \ref{Prop P
properties} and Theorem \ref{Theorem triviality} provide the
existence of a Bloch frame which is even real-analytic. While this
extra regularity is unessential  for problem $\mathrm{(MV_1)}$, it
will be crucial when dealing with $\mathrm{(MV_2)}$. Notice that, under
the previous  assumption, the set of  admissible
Bloch frames in Definition \ref{Def Optimally loc Wannier functions}
is non-void, so problem $\mathrm{(MV_1)}$ makes sense.

By using the previous Bloch frame, we can lift the functional
\eqref{mvfimbande} to $W^{1,2}$-maps from $\Base$ to the unitary
group $\mathcal{U}(m)$, \ie to $\Gamma^*$-periodic maps from $\R^d$
to $\Un(m)$. Indeed, given any map $U \in W^{1,2}(\Base, \Un(m))$
one defines a Bloch frame $\ph = \{\ph_1, \ldots, \ph_m \} \subset
\tilde{\W}$ by setting $\ph = \chi \cdot U$, \ie $\ph_a(k,\cdot) =
\sum_b \chi_b(k, \cdot) \, U_{b,a}(k)$. Vice versa, if $\ph =
\{\ph_1, \ldots, \ph_m \} \subset \tilde{\W}$ is a Bloch frame, then
pointwise $\ph_a(k,\cdot) = \sum_b \chi_b(k, \cdot) \, U_{b,a}(k)$
with $U_{b,a}(k) = \inner{\chi_b(k)}{\ph_a(k)}$, hence $U \in
W^{1,2}(\Base, \Un(m))$.

\noindent For the given reference frame $\chi$, the functional
\eqref{mvfimbande} in terms of the gauge $U$ becomes
\begin{eqnarray} \label{MVU}
 \tilde{F}_{MV}(U;\chi) &=& \sum_{j=1}^d
\int_{\mathbb{T}^*_d} \left[ \tr \left( \frac{\partial
U^{*}}{\partial k_j}(k) \frac{\partial U}{\partial k_j}(k) \right)
+m  \sum_{a=1}^m \left\| \frac{\partial \chi_a(k, \cdot)}{\partial
k_j}\right\|^2_{\mathcal{H}_f} \right] dk \,\, + \\
\nonumber 
&+&  \sum_{j=1}^d  \int_{\mathbb{T}_d^*} \tr \left[ \left(
U(k)\frac{\partial U^*}{\partial k_j}(k)- \frac{\partial U}{\partial
k_j}(k) U^*(k) \right) A_j(k) \right] dk \, + \\
\nonumber 
&+&  \sum_{a=1}^m \sum_{j=1}^d   \left(   \int_{\mathbb{T}^*_d}
\left[ U^*(k) \left( \frac{\partial U}{\partial k_j}(k) + A_j(k)U(k)
\right) \right]_{aa} dk \right)^2 \,.
\end{eqnarray}
Here the matrix coefficients $A_j \in L^2(\mathbb{T}^*_d;
\mathfrak{u}(m))$ are given by the formula
\begin{equation}
\label{berry} \big[A_j(k)\big]_{cb}= \inner{
\chi_c(k,\cdot)}{\frac{\partial \chi_b(k,\cdot)}{\partial
k_j}}_{\mathcal{H}_f}
\end{equation}
When $\chi$ is real-analytic, the functions $A_j \in
C^{\omega}(\mathbb{T}^*_d; \mathfrak{u}(m))$  represent the
antihermitian connection $1$-form induced on the (sub)bundle $\E_*$
by the trivial connection on the bundle $\mathbb{R}^d \times
\mathcal{H}_f$.

\noindent Moreover,
\begin{equation}\label{Inf equality}
\inf \set{F_{MV}(w) \,: \, \begin{array}{c}
                            \set{w_1, \ldots, w_m} \subset \W \\[2mm]
                            \UZ \, w \mbox{ is a Bloch frame}
                         \end{array}} =
\inf \set{\tilde{F}_{MV}(U;\chi) \,: \, U \in W^{1,2}(\Base;
\Un(m))}.
\end{equation}
Therefore, problem $\mathrm{(MV_1)}$ is equivalent to showing that the
r.h.s.\ of (\ref{Inf equality}) is attained. Analogously, in view of
Corollary \ref{Cor Exp local},  problem $\mathrm{(MV_2)}$
corresponds to show that any minimizer of
$\tilde{F}_{MV}(\cdot;\chi)$ is real-analytic, provided that $\chi$
is also real-analytic.

\begin{ciao}[\textbf{Rough reference frames, as in numerical simulations}]\label{Rem rough references}
When reformulating problem $\mathrm{(MV_1)}$ in terms of the
functional $\tilde{F}_{MV}( U ;\chi)$ the regularity of the
reference frame $\chi$ plays no essential role, provided $\chi$ is
in $\tilde\W$. Indeed, in view of (\ref{Inf equality}), the infimum
on the r.h.s does not depend on the choice of $\chi$. Moreover, if
for a particular choice of $\chi$ the infimum is attained at $U$,
then for another choice $\tilde \chi \subset \tilde \W$ the infimum
is attained at $\tilde U = V \, U$, where $\chi = \tilde \chi \cdot
V$. The matrix $V$, defined by $V_{b,a}(k) =
\inner{\tilde\chi_b(k)}{\chi_a(k)}$, is in $W^{1,2}(\Base; \Un(m))$,
hence $\tilde U $ is also in $W^{1,2}(\Base; \Un(m))$.
This observation justifies the fact that a
minimizer (whenever it exists) can be evaluated starting from any reference Bloch frame $\chi \subset \tilde \W$,
even a discontinuous one, as it happens in numerical simulations.
\end{ciao}

\begin{ciao}[\textbf{Minimizing over smooth gauges}]
\label{Rem smooth gauges} To compute the infimum of
$\tilde{F}_{MV}(U; \chi)$ it is sufficient to consider smooth change
of gauges. More precisely, for $d \leq 3$ and for any fixed Bloch
frame $\chi \subset \tilde\W$ one has
$$
\inf \set{\tilde{F}_{MV}(U; \chi)\,:\  U \in W^{1,2}(\Base; \Un(m))}
= \inf \set{\tilde{F}_{MV}(U; \chi)\,:\  U \in C^{\infty}(\Base;
\Un(m))}.
$$
As a consequence, in numerical implementations, to compute the above
infimum one can let $U$ vary in any set $S$ such that
$C^{\infty}(\Base; \Un(m)) \subset S \subset W^{1,2}(\Base;
\Un(m))$.

\noindent This result follows from  the strong density of smooth
maps $C^{\infty} (\mathbb{T}^*_d;\Un(m)) \subset W^{1,2}
(\mathbb{T}^*_d; \Un(m))$ for $d \leq 3$. If $d=2$ the latter claim
is essentially the approximation by convolution followed by the
nearest-point projection onto the unitary group; if $d=3$, the claim
follows from the fact that the homotopy group $\pi_2(\mathcal{U}(m))$ is
trivial and from the fact that $\mathbb{T}^*_3$ has the 1-extension
property with respect to $\mathcal{U}(m)$, see \cite[Theorem 1.3
and Section 5]{HL}.
\end{ciao}


\goodbreak

\subsection{Existence of minimizers}

The following results shows that the right hand side in \eqref{Inf
equality} is attained. The proof is a simple modification of the
direct method in the calculus of variations,  in order to handle a
natural invariance of the functional \eqref{mvwmbande}. Indeed, if
$ \{w_1, \ldots, w_m \}$ are composite Wannier functions satisfying
\eqref{constraint} and $\{\gamma_1, \ldots, \gamma_m\}\subset \Gamma
$, then
\begin{equation}
\label{wannierinvariance}
F_{MV}(w_1,\ldots,w_m)=F_{MV}(\tilde{w}_1,\ldots,\tilde{w}_m) \, ,
\quad \tilde{w}_a(x)=w_a(x+\gamma_a) \, , \quad 1\leq a\leq m \, .
\end{equation}
Moving to Bloch functions, we have
$\displaystyle{\tilde{\varphi}_a(k, \cdot)\equiv \left(
\tilde{\mathcal{U}}_{BF} \, \tilde{w}_a \right)(k, \cdot)=
e^{ik\gamma_a}\left( \tilde{\mathcal{U}}_{BF} \, w_a \right)(k,
\cdot) }$, so that  $\{ \tilde{\varphi}_1(k, \cdot),\ldots,
\tilde{\varphi}_m(k, \cdot)\}$ is still orthonormal in
$\mathcal{H}_f$ and \eqref{constraint} holds.  Correspondingly, the
functional \eqref{MVU} has the invariance
\begin{equation}
\label{Uinvariance}
\tilde{F}_{MV}(U;\chi)=\tilde{F}_{MV}(\tilde{U};\chi) \, , \quad
\tilde{U}(k)= \diag \left( e^{ik\gamma_a}\right) U(k) \, .
\end{equation}

\begin{theorem}
\label{existence} Let $\{ P_*(k) \}_{k \in \R^d} \subset \B(\Hf)$ be
a family of orthogonal projectors satisfying properties
$\mathrm{(\tilde{P}_1)}$ and $\mathrm{(P_2)}$, with $\dim P_*(k)
\equiv m$. Assume that there exists a Bloch frame $\chi \subset
\tilde \W$. Then there exists $U \in
W^{1,2}(\mathbb{T}^*_d;\mathcal{U}(m))$ which is a minimizer on
$W^{1,2}(\mathbb{T}^*_d;\mathcal{U}(m))$ of the localization
functional $\tilde{F}_{MV}(\cdot, \chi)$ defined by (\ref{MVU}).
\end{theorem}

\begin{proof}
Note that $\Tilde{F}(\cdot;\chi)$ is nonnegative, because of
\eqref{Inf equality} and \eqref{mvwmbande}. Note also that
$E(U)=\sum_{j=1}^d \int_{\mathbb{T}^*_d} \frac12   \tr \left(
\frac{\partial U^{*}}{\partial k_j}(k) \frac{\partial U}{\partial
k_j}(k) \right)$ is the standard Dirichlet integral so it is
(sequentially) weakly lower semicontinuous on $W^{1,2}$. Thus
$\Tilde{F}(\cdot;\chi)$ is weakly lower semicontinuous in $W^{1,2}$,
because the first term in \eqref{MVU} is, up to a constant factor,
$E(\cdot)$ and the other terms are clearly weakly continuous
(because of the compact embedding $W^{1,2}(\mathbb{T}^*_d)
\hookrightarrow L^2(\mathbb{T}^*_d)$). In order to apply the direct
method it remains to show that the functional is coercive, so that
there exists a minimizing sequence which is bounded in
$W^{1,2}(\mathbb{T}^*_d;\mathcal{U}(m))$. However, coercivity
clearly fails because of the natural invariance of
$\Tilde{F}(\cdot;\chi)$ given by \eqref{Uinvariance} (the
$W^{1,2}$-norm of $\tilde{U}$ can be made arbitrarily large as
$|\gamma_a|\to \infty$).

In order to fix the argument it is enough to take advantage of this
invariance in the form \eqref{wannierinvariance}. Indeed, for any
admissible $U$ we can choose $\gamma_1, \ldots,\gamma_m$ so that the
corresponding $\tilde{U}$, defined as in \eqref{Uinvariance}, gives
the same value of $\tilde{F}_{MV}(\cdot;\chi)$ but in such a way
that the corresponding Wannier functions $\{ w_a\}$ satisfy $
\sum_{j=1}^d \left|\int_{\mathbb{R}^d} x_j |w_a(x)|^2dx\right| \leq
C$ for some absolute constant $C>0$ independent of $U$. Thus, if
along a sequence $\{U_n \}$ the functional $\tilde{F}_{MV}$ is
bounded, up to a suitable choice of the translation parameters $\{
\gamma_a^n\}$ we have a uniform bound for the modified sequence
$\{\tilde{U}_n\}$ in $W^{1,2}$ (the third line in \eqref{MVU} is now
bounded by $C^2 m d$ and the first two are easily seen to be
equivalent to the Dirichlet energy up to additive and multiplicative
constants depending only on $A$, $C$ and $m$).

Then, up to subsequences, $\tilde{U}_n$ weakly converges to some $U
\in W^{1,2}(\mathbb{T}^*_d;\mathcal{U}(m))$ which is a minimizer.
\end{proof}

\goodbreak

\subsection{The Euler-Lagrange equations}
In this subsection we derive the Euler-Lagrange equations
corresponding to the functional \eqref{MVU}. To our knowledge, these equations
are new in the literature. We consider the unitary
group $\mathcal{U}(m)$ as isometrically embedded into the space $M_m
(\mathbb{C})$ with the standard real Euclidean product
$\inner{A}{B}= \re \tr (A^* B)$. Recall that $\Un(m)$ is a compact
real Lie group, the induced metric is biinvariant, its Lie algebra
is given by the real vector space of complex antihermitian matrices
and it has dimension $m^2$.

Let $\varphi \in C^\infty (\mathbb{T}_d^*; M_m(\mathbb{C}))$ and for
$\varepsilon\neq 0$ fixed let $U(k)+\varepsilon \varphi(k)$ be a
free variation of $U$ in the direction $\varphi$. In a sufficiently
small tubular neiborhood $\mathcal{O}$ of $\mathcal{U}(m)$ in
$M_m(\mathbb{C})$ there is a well defined smooth nearest point projection
map $\Pi:\mathcal{O} \to \mathcal{U}(m)$, so we can consider the
induced variations

\begin{equation}
\label{variations} U_\varepsilon(k) := \Pi(U(k)+\varepsilon
\varphi(k))=U(k)\left( \mathbb{I} + \varepsilon \frac12 \left[
U^{-1}(k)\varphi(k)- (U^{-1}(k)\varphi(k))^* \right] \right)
+o(\varepsilon) \, .
\end{equation}
Simple calculations on each term in \eqref{MVU} yield

\begin{equation}
\label{derivata1} \left. \frac{d}{d \varepsilon}
\right\vert_{\varepsilon=0} \int_{\mathbb{T}^*_d}   \tr \left(
\frac{\partial U_{\varepsilon}^*}{\partial k_j}(k) \frac{\partial
U_{\varepsilon}}{\partial k_j}(k) \right) dk  = 2
\int_{\mathbb{T}_d^*} \, \tr \left[  \frac{\partial
\varphi^*}{\partial k_j}(k) \frac{\partial U}{\partial k_j}(k)+
\varphi^*(k)  \frac{\partial U}{\partial k_j}(k) U^{-1}(k)
\frac{\partial U}{\partial k_j}(k) \right] dk   \, ,
\end{equation}

\begin{equation}
\label{derivata2} \left. \frac{d}{d \varepsilon}
\right\vert_{\varepsilon=0} \int_{\mathbb{T}_d^*} \tr \left[ \left(
U_\varepsilon (k)\frac{\partial U_\varepsilon^*}{\partial k_j}(k)-
\frac{\partial U_\varepsilon}{\partial k_j}(k) U_\varepsilon^*(k)
\right) A_j(k) \right] dk =
\end{equation}
\[ = 2  \int_{\mathbb{T}_d^*} \, \tr \left[  \varphi^*(k)
\left\{  \frac{\partial U}{\partial k_j}(k) U^{-1}(k) A_j(k) U(k) -
\frac{\partial A_j}{\partial k_j}(k)U(k)- A_j(k) \frac{\partial
U}{\partial k_j}(k) \right\} \right] \, dk                   \, ,
\]
\noindent and
\begin{equation}
\label{derivata3}
\left. \frac{d}{d \varepsilon}  \right\vert_{\varepsilon=0} \sum_{a=1}^m    \left(   \int_{\mathbb{T}^*_d}  \left[ U_\varepsilon^*(k) \left( \frac{\partial U_\varepsilon}{\partial k_j}(k) +  A_j(k)U_\varepsilon(k) \right) \right]_{aa} dk \right)^2=
\end{equation}
\[
= 2  \int_{\mathbb{T}_d^*} \, \tr \left[  \varphi^*(k) \left\{
-\left( \frac{\partial U}{\partial k_j}(k) + A_j(k) U(k) \right)
G^j+ U(k) G^j U^{-1}(k) \left(\frac{\partial U}{\partial k_j}(k) +
A_j(k) U(k) \right) \right\} \right] dk                   \, .
\]

\noindent Here the constant (purely imaginary) diagonal matrices $\{
G^j \} \subset M_m(\mathbb{C})$ are defined as $ G^j= \diag \, \(
\int_{\mathbb{T}_d^*} U^*(k) \left[ \frac{\partial U}{\partial
k_j}(k) + A_j(k) U(k) \right] dk \)$, where $\left[ \diag M \right]_{ab} = M_{aa} \delta_{ab}$.

Thus a map $U \in W^{1,2}(\mathbb{T}_d^*;\mathcal{U}(m))$ satisfies
$\displaystyle{\frac{d}{d\varepsilon}
\tilde{F}_{MV}(U_{\varepsilon}; \chi) {\vert}_{\varepsilon=0} =0}$
if and only if $U$ is a  weak solution to the Euler-Lagrange
equations

\begin{equation}
\label{EuleroLagrange}
-\Delta U+ \sum_{j=1}^d \frac{\partial U}{\partial k_j} U^{-1} \frac{\partial U}{\partial k_j} + \sum_{j=1}^d \left[  \frac{\partial U}{\partial k_j} U^{-1} A_j U - \frac{\partial A_j}{\partial k_j}U-A_j \frac{\partial U}{\partial k_j} \right]+
\end{equation}
\[ + \sum_{j=1}^d \left[ -\left( \frac{\partial U}{\partial k_j} + A_j U \right) G^j+ U G^j U^{-1}
\left(\frac{\partial U}{\partial k_j} + A_j U \right) \right]=0 \,.
\]

\section{Continuity of minimizers}
The goal of this section is to show that a change of gauge $U\in
W^{1,2}(\mathbb{T}_d^*;\mathcal{U}(m))$ which minimizes the functional
\eqref{MVU} is continuous, whenever $\chi$ is real-analytic. This
result will be true if $m=1$ for any $d$, if $1\leq d \leq 2$ for
any $m$ and for $d=3$ under the restriction $2\leq m\leq3$. The
proofs in the three cases are different. The first case, treated in
the next proposition,  is the simplest and it gives even
real-analiticity. In order to deal with the  other two cases,  we
will need a dimension dependent argument which will occupy the rest
of the section. Note that in the first two cases continuity holds
for any solution of the Euler Lagrange equations. In contrast,
in the three dimensional case continuity relies on energy minimality
in an essential way.

To deal with the regularity of the minimizers, we make the following
assumption which, as already noticed after stating the weaker Assumption 1, is
automatically satisfied for $d \leq 3$ or $m=1$ (compare Theorem \ref{Theorem triviality}).

\noindent \textbf{Assumption 2:} there exists a Bloch frame $\chi=\{
\chi_1,\ldots, \chi_m\} \subset \Ht$ such that $\chi_a \in
C^{\omega}(\R^d, \Hf)$ for every $a \in \set{1, \ldots,m}$.

\begin{proposition}
\label{regsingleband} Let $d=1$ and $m \geq 1$, or $d\geq 1$ and
$m=1$. Let $U\in W^{1,2}(\mathbb{T}_d^*;\mathcal{U}(m))$ be a weak
solution to equation \eqref{EuleroLagrange}. Then $U$ is
real-analytic.
\end{proposition}

\begin{proof}
Assume $d=1$. Since $U \in W^{1,2}(\Base, \Un(m))$, then
(\ref{EuleroLagrange}) yields $\frac{d^2}{dk^2}U \in L^1$. Thus $U
\in W^{2,1}$. Therefore $U$ is in $C^1$ by Sobolev embedding, and
(\ref{EuleroLagrange}) implies $U \in C^2$. Analogously, if $U \in
C^{n+1}$ solves (\ref{EuleroLagrange}), $n \geq 1$, then $U \in
C^{n+2}$, so by induction $U \in C^{\infty}$. Finally, since $A_j
\in C^{\omega}$ for all $j$, by the standard ODE regularity theory
one obtains $U \in C^{\omega}$.

Now assume $m=1$. Since $\mathcal{U}(1)$ is abelian the equation
\eqref{EuleroLagrange} reduces to
\begin{equation}
\label{ELabelian}
-\Delta U+ \sum_{j=1}^d \frac{\partial U}{\partial k_j} U^{-1} \frac{\partial U}{\partial k_j}= \left( \sum_{j=1}^d \frac{\partial A_j}{\partial k_j} \right) U \, .
\end{equation}
Recall that in any sufficiently small ball $B\subset \mathbb{T}_d^*$
we have $U(k)=e^{if(k)}$ for some $f\in W^{1,2}(B;\mathbb{R})$ (see
\cite{BZ}). Thus, equation \eqref{ELabelian} reads $\Delta
f=i \sum_{j=1}^d \frac{\partial A_j}{\partial k_j} \in
C^{\omega}(B;\mathbb{R})$, because the matrices in the Berry
connection are antihermitian and real-analytic. Since the Laplacian
is analytic-hypoelliptic we conclude $f \in C^{\omega}(B)$ and in
turn $U\in C^{\omega}(\mathbb{T}_d^*; \mathcal{U}(1))$ since the
ball $B$ can be choosen arbitrarily.
\end{proof}

\subsection{Continuity in the two dimensional case}
We are going to prove continuity of any weak solutions to \eqref{EuleroLagrange} in the case $d=2$. The argument here is just sketched, since, up to a standard localization argument, it essentially follows the proof of continuity for weakly harmonic maps from a two dimensional domain into spheres (see e.g. \cite{LW}, Chapter III, Section 3.2 pag. 57-61).

We start with the following auxiliary result which is a straightforward consequence of \cite{CLMS}.
\begin{lemma}
\label{linearregularity} Let $d\geq 2$, $m\geq 2$ and let  $U\in
W^{1,2}(\mathbb{T}_d^*;\mathcal{U}(m))$.   Assume  $\tilde{B}^j \in L^{2}(\mathbb{T}_d^*; \mathfrak{u}(m))$ for
$j \in \set{1,\ldots, d}$ and $ \dive \tilde{B}=0$ in $\mathcal{D}^\prime
(\mathbb{T}^*_d)$.  If $U$ is a weak solution to
\begin{equation}
\label{linearharmonic}
\Delta U= \sum_{j=1}^d  \frac{\partial U}{\partial k_j} \tilde{B}^j+\tilde{f} \qquad  \mbox{and $U^{-1}\tilde{f}\in L^p(\mathbb{T}^*_d; \mathfrak{u}(m))$}
\end{equation}
for some $p>1$, then $U\in W^{2,1}(\mathbb{T}_d^*;\mathcal{U}(m))$.
\end{lemma}
\begin{proof}
We have
$$
\Delta U_{a,c}=\sum_{b=1}^m \sum_{j=1}^d  \frac{\partial U_{ab}}{\partial k_j}\tilde{B}^j_{bc} +\tilde{f}_{ac} =\sum_{b=1}^m  E_{ab} \cdot \tilde{B}_{ab}+\tilde{f}_{ac} \, ,
$$
where for fixed $a,b,c$ the vector fields $\tilde{B}_{bc}$ and $E_{ab}=\nabla U_{ab}$ are in $L^2$ and satisfy $\dive  \tilde{B}_{bc}=0$ and $\curl E_{ab}=0$ in the sense of distributions. According to \cite{CLMS}, we have $  E_{ab} \cdot \tilde{B}_{bc} \in \mathcal{H}_{\rm loc} ^1(\mathbb{T}^*_d)$ and $\tilde{f}_{ac} \in \mathcal{H}_{\rm loc} ^1(\mathbb{T}^*_d)$, \ie the right hand side of \eqref{linearharmonic} is in the local Hardy space $\mathcal{H}^1_{\rm loc} (\mathbb{T}^*_d) \subset L^1(\mathbb{T}^*_d)$. Here $g \in \mathcal{H}^1_{\rm loc} (\mathbb{T}^*_d)$ means that $g \in \mathcal{H}^1_{\rm loc}(B)$ for every sufficiently small ball $B \subset \T_d^*$ (namely, for every ball with radius smaller than the injectivity radius of the exponential map). For the definition and the basic properties of $\mathcal{H}^1_{\rm loc}(\R^d)$ we refer to \cite[Chapter 3]{St}.
Thus the conclusion follows from \cite[Theorem 3.2.4]{LW}.
\end{proof}

Based on the previous lemma we have the following intermediate result.

\begin{proposition}
\label{approxhmreg} Let $d\geq 2$, $m\geq 2$ and $U\in
W^{1,2}(\mathbb{T}_d^*;\mathcal{U}(m))$.  If $U$ is a weak solution to
\begin{equation}
\label{approxharmonic}
\Delta U= \sum_{j=1}^d \frac{\partial U}{\partial k_j} U^{-1} \frac{\partial U}{\partial k_j} +f
\qquad \mbox{and $U^{-1} f \in L^2(\mathbb{T}^*_d; \mathfrak{u}(m))$,}
\end{equation}
then $U\in W^{2,1}(\mathbb{T}_d^*;\mathcal{U}(m))$. In particular, if $d=2$ then $U\in C^0(\mathbb{T}_2^*;\mathcal{U}(m))$.
\end{proposition}
\begin{proof}
If we set $B^j=\frac{1}{2}\left( U^* \frac{\partial U}{\partial k_j}
- \frac{\partial U^*}{\partial k_j}U \right)$ then $ B^j \in
L^{2}(\mathbb{T}_d^*; \mathfrak{u}(m))$ for $j \in \set{1, \ldots, d}$, $U$ is
a weak solution to \eqref{linearharmonic} but $\dive B=
\frac{1}{2}(U^*f- f U^*) \neq 0$ in $\mathcal{D}^\prime
(\mathbb{T}^*_d)$. Note that $\dive B \in L^2_0(\mathbb{T}_d^*;
\mathfrak{u}(m))$, the space of zero-mean $L^2$-integrable functions, hence $\Delta^{-1} \, \dive B \in
W^{2,2}(\mathbb{T}_d^*; \mathfrak{u}(m))$ by elliptic regularity and
if we set
\[ \tilde{B}= B- \nabla \Delta^{-1} \dive B \, , \qquad  \tilde{f}=f+ \nabla U \cdot \nabla \Delta^{-1} \dive B \, ,  \]
then $\tilde{B}$, $\tilde{f}$ and $U$ satisfy the assumptions of Lemma \ref{linearregularity} by Sobolev embedding (for some $p>1$ depending only on $d$), hence we obtain $U\in W^{2,1}(\mathbb{T}_d^*;\mathcal{U}(m))$. Finally, if $d=2$ the improved Sobolev embeddings into Lorentz spaces yield $\nabla U \in L^{2,1}$ and in turn $U \in C^0$ (see \cite{LW}, Theorem 3.2.7 and 3.2.8 respectively).
\end{proof}

Going back to equation \eqref{EuleroLagrange} we have the  following
important consequence
\begin{corollary}
\label{2dcontinuity}
Let $d=2$, $m\geq 2$ and $U\in W^{1,2}(\mathbb{T}_2^*;\mathcal{U}(m))$ a weak solution to equation \eqref{EuleroLagrange}. Then $U\in W^{2,1}(\mathbb{T}_2^*;\mathcal{U}(m))$ and $U\in C^0(\mathbb{T}_2^*;\mathcal{U}(m))$.
\end{corollary}
\begin{proof}
If $U\in W^{1,2}(\mathbb{T}_2^*;\mathcal{U}(m))$ is a weak solution to
equation \eqref{EuleroLagrange} then $U$ satisfies also
\eqref{approxharmonic} for a suitable $f$ depending on $U$. Here
$U^{-1} f \in L^2(\mathbb{T}^*_2; \mathfrak{u}(m))$ because of the
regularity property of $U$ and the Berry connection $A_j \in C^{\omega}(\mathbb{T}_2^*,\mathfrak{u}(m))$,
in view of the structure of equation \eqref{EuleroLagrange}. Thus, the conclusion
follows from Proposition \ref{approxhmreg}.
\end{proof}

\subsection{Continuity in the three dimensional case}
The goal of this subsection is to show that if $d=3$ then minimizers
of the localization functional \eqref{MVU} are continuous, at least
if $m\leq3$. Roughly speaking the idea is to prove that at smaller
and  smaller scales minimizers look like minimizing harmonic maps
into the unitary group $\mathcal{U}(m)$ which are degree-zero
homogeneous. Since the latter are constant, at least for $m\leq3$ (see Corollary \ref{liouville} below),
then the former are continuous at a sufficiently small scale. All
the techniques in this section are inspired by the regularity theory
for minimizing harmonic maps (see \cite{SU1},  \cite{S} and
\cite{LW}).

The first condition we need to study minimizers at small scales is the stationarity condition with respect
to inner variations.

\begin{lemma}
\label{innervariations} Let $d \geq 3$ and let $U\in
W^{1,2}(\Base;\mathcal{U}(m))$ be a minimizer of \eqref{MVU}. For
each $1\leq j\leq d$ define constant diagonal matrices $G^j= \diag
\int U^*(\frac{\partial U}{\partial k_j}+A_j U) dk$. Let $\Phi\in
C^\infty(\mathbb{T}^*_d;\mathbb{R}^d)$ be a smooth vector field and
$\Psi_\varepsilon(k)=k+\varepsilon \, \Phi(k)$  be a family of
diffeomorphisms (for $\varepsilon$ small enough). Then

\[ \left. \frac{d}{d\varepsilon} \right|_{\varepsilon=0} \tilde{F}_{MV}(U\circ \Psi_\varepsilon ; \chi)=\sum_{j,c}\int \frac{\partial \Phi^c}{\partial k_c} \tr \left[ -\frac{\partial U^*}{\partial k_j}\frac{\partial U}{\partial k_j}+2 \frac{\partial U}{\partial k_j}U^*A_j-2 G^jU^* \left(\frac{\partial U}{\partial k_j}+A_j U\right) \right]+\]
\[ \sum_{j,c} \int \frac{\partial \Phi^c}{\partial k_j} \tr \left[ \frac{\partial U^*}{\partial k_c}\frac{\partial U}{\partial k_j} +\frac{\partial U^*}{\partial k_j}\frac{\partial U}{\partial k_c}-2 \frac{\partial U}{\partial k_c}U^*A_j+2 G^j U^* \frac{\partial U}{\partial k_c} \right] + \]
\begin{equation}
\label{stationary}
 \sum_{j,c} \int  \Phi^c \tr \left[ 2 \frac{\partial U}{\partial k_j} \frac{\partial A_j}{\partial k_c}-2 G^j U^* \frac{\partial A_j}{\partial k_c}U \right] \, .
 \end{equation}
\end{lemma}

\begin{proof}
First note that $\det D\Psi_\varepsilon^{-1}(k)=1-\varepsilon \dive \Phi+o(\varepsilon)$ uniformly on $\mathbb{T}^*_d$. As a  consequence the change of variable and the chain rule give
\begin{equation}
\label{auxstationary}
\left. \frac{d}{d\varepsilon}\right|_{\varepsilon=0} \int f(\Psi_\varepsilon(k))g(k)dk=-\int f(k) g(k) \dive\Phi(k) dk - \sum_{c} \int f(k)\frac{\partial g}{\partial k_c} (k) \Phi^c(k) dk \, .
\end{equation}
for any $f \in L^1$ and for any $g\in C^1$.
Now we set $U_\varepsilon(k)=U\circ\Psi_\varepsilon (k)$, so that
\begin{equation}
\label{auxstationary2}
 \frac{\partial U_\varepsilon}{\partial k_j}(k)=\sum_c \frac{\partial U}{\partial k_c}(\Psi_\varepsilon(k))\frac{\partial \Psi_{\varepsilon}^c (k)}{\partial k_j}= \sum_c \frac{\partial U}{\partial k_c}(\Psi_\varepsilon(k))(\delta^c_j+\varepsilon \frac{\partial \Phi^c(k)}{\partial k_j}) \, .
 \end{equation}
Thus, a repeated application of \eqref{auxstationary} and \eqref{auxstationary2} gives
\begin{equation}
\label{stationary1}
\left. \frac{d}{d\varepsilon}\right|_{\varepsilon=0} \sum_j \int \tr\left[ \frac{\partial U^*_\varepsilon}{\partial k_j}\frac{\partial U_\varepsilon}{\partial k_j}\right] dk=- \sum_j \int \dive \Phi(k) \, \tr\left[ \frac{\partial U^*_\varepsilon}{\partial k_j}\frac{\partial U_\varepsilon}{\partial k_j}\right]+
\end{equation}
\[ \sum_{j,c} \int \frac{\partial \Phi^c(k)}{\partial k_j}\tr\left[ \frac{\partial U^*}{\partial k_c}\frac{\partial U}{\partial k_j} +\frac{\partial U^*}{\partial k_j}\frac{\partial U}{\partial k_c} \right]\, , \]
\begin{equation}
\label{stationary2}
\left. \frac{d}{d\varepsilon}\right|_{\varepsilon=0} \sum_j \int \tr \left[ \left( U_\varepsilon \frac{\partial U^*_\varepsilon}{\partial k_j}-\frac{\partial U_\varepsilon}{\partial k_j}U^*_\varepsilon \right) A_j(k) \right] dk=
\end{equation}
\[ =- \sum_{j,c} \int  \, \tr \left[ \left( U \frac{\partial U^*}{\partial k_j}-\frac{\partial U}{\partial k_j}U^* \right) \left( A_j(k) \frac{\partial \Phi^c}{\partial k_c}+ \frac{\partial A_j}{\partial k_c}\Phi^c \right) \right] dk +\]
\[+\sum_{j,c} \int \tr \left[ \left( U \frac{\partial U^*}{\partial k_c}-\frac{\partial U}{\partial k_c}U^* \right) A_j(k) \frac{\partial \Phi^c}{\partial k_j} \right] dk .  \]
Similarly, taking the definition of $G^j$ into account, one has
\begin{equation}
\label{stationary3}
\left. \frac{d}{d\varepsilon}\right|_{\varepsilon=0} \sum_{j, a} \left( \int \left[ U^*_\varepsilon \left(\frac{\partial U_\varepsilon}{\partial k_j}+A_j U_\varepsilon \right) \right]_{aa} \right)^2=
\end{equation}
\[= \sum_{j,c} \int \tr\left[- 2 G^j  U^* \left( \frac{\partial U}{\partial k_j}+A_jU \right) \frac{\partial \Phi^c}{\partial k_c}+2 G^j U^* \frac{\partial U}{\partial k_c} \frac{\partial \Phi^c}{\partial k_j}-2G^jU^* \frac{\partial A_j}{\partial k_c} U \Phi^c \right] \, . \]
Combining \eqref{stationary1}, \eqref{stationary2} and \eqref{stationary3} with the definition of \eqref{MVU} and reorganizing the sum we easily have \eqref{stationary}.
\end{proof}

The first consequence we obtain is a sort of perturbed monotonicity formula (in the spirit of the monotonicity formula for almost harmonic maps; see \cite{Mos}, Chapter 4).

\begin{proposition}
\label{monotonicity} Let $d\geq3$ and let $U\in
W^{1,2}(\mathbb{T}_d^*;\mathcal{U}(m))$ be a minimizer of the
localization functional \eqref{MVU}. Then there exist $C>0$ and
$\bar{R}>0$ such that for each $k_0\in \mathbb{T}^*_d$ and for each
$R_1, R_2$ with $0<R_1\leq R_2 \leq \bar{R}$ we have
\begin{equation}
\label{almostmon}
 \frac{1}{R_1^{d-2}}\int_{B_{R_1}(k_0)} \!\!\!\!\!\! |\nabla U|^2+\int_{R_1<|k-k_0|<R_2} \frac{2}{|k-k_0|^{d-2}} \left|\frac{\partial U}{\partial r}\right|^2 \leq 2CR_2+(1+ 2CR_2) \frac{1}{R_2^{d-2}} \int_{B_{R_2}(k_0)} \!\!\!\!\!\! |\nabla U|^2  \, . \end{equation}

\noindent As a consequence, the function $\displaystyle{R^{2-d} \int_{B_R(k_0)}  | \nabla U|^2dk}$ has a limit and $ \int_{B_R(k_0)}  | \partial_r U|^2|k-k_0|^{2-d}dk $ is finite and vanishes as $R \to 0^+$.
\end{proposition}
\begin{proof}
Clearly by translation invariance it suffices to consider the case $k_0=0$. We take $h\in C^\infty(\mathbb{R})$ an increasing function such that $h(t)\equiv 0$ for $t<0$ and $h(t)\equiv 1$ for $t\geq 1$. If we take $\Phi(k)=k \, h(R-|k|)$ in \eqref{stationary} we have $\dive \Phi (k)=d \, h(R - |k|) -|k| h^\prime(R - |k|)$ ($d \in \N$ is the dimension) and
\begin{equation}
\label{auxstationary3}
\frac{\partial \Phi^c}{\partial k_j}=\delta^c_j h-\frac{k_ck_j}{|k|}h^\prime \, ,  \, \, \sum_c \frac{\partial U}{\partial k_c}\frac{\partial \Phi^c}{\partial k_j}= \frac{\partial U}{\partial k_j} h-k_j \frac{\partial U}{\partial r} h^\prime \, , \, \, \sum_c \frac{\partial A_j}{\partial k_c}\Phi^c=|k| \frac{\partial A_j}{\partial r} h \, .
\end{equation}
Combining \eqref{stationary} and \eqref{auxstationary3} we get
\[ 0=\sum_{j}\int_{B_R} \left(d h-|k| h^\prime\right)  \tr \left[ -\frac{\partial U^*}{\partial k_j}\frac{\partial U}{\partial k_j}+2 \frac{\partial U}{\partial k_j}U^*A_j-2 G^jU^* \left(\frac{\partial U}{\partial k_j}+A_j U\right) \right]+\]
\[ -\sum_{j} \int_{B_R} k_j \tr \left[ \frac{\partial U^*}{\partial r}\frac{\partial U}{\partial k_j} +\frac{\partial U^*}{\partial k_j}\frac{\partial U}{\partial r}-2 \frac{\partial U}{\partial r}U^*A_j+2 G^j U^* \frac{\partial U}{\partial r} \right] h^\prime + \]
\begin{equation}
\label{premonotonicity1}
+2\sum_{j} \int_{B_R}    \tr \left[ \frac{\partial U^*}{\partial k_j}\frac{\partial U}{\partial k_j} - \frac{\partial U}{\partial k_j}U^*A_j+ G^j U^* \frac{\partial U}{\partial k_j} +  |k|\frac{\partial U}{\partial k_j} \frac{\partial A_j}{\partial r} - |k| G^j U^* \frac{\partial A_j}{\partial r}  U \right] h  \, .
 \end{equation}
 Passing to the limit as $h\to \chi_{\{t>0\}}$, as in \cite[Chapter 2]{S}, yields
 \begin{equation}
 \label{premonotonicity2}
  0=\sum_{j}\int_{B_R}   \tr \left[ (2-d) \frac{\partial U^*}{\partial k_j}\frac{\partial U}{\partial k_j} \right]+R \sum_j \int_{\partial B_R} \tr \left[ \frac{\partial U^*}{\partial k_j}\frac{\partial U}{\partial k_j} \right] - 2R \int_{\partial B_R} \tr \left[ \frac{\partial U^*}{\partial r}\frac{\partial U}{\partial r} \right]
 \end{equation}

 \[ -\sum_{j}\int_{\partial B_R} |k|   \tr \left[ 2 \frac{\partial U}{\partial k_j}U^*A_j-2 G^jU^* \left(\frac{\partial U}{\partial k_j}+A_j U\right)  -2 \frac{k_j}{|k|}\frac{\partial U}{\partial r}U^*A_j+2 \frac{k_j}{|k|}G^j U^* \frac{\partial U}{\partial r} \right]+\]

 \[ 2\sum_{j} \int_{B_R}    \tr \left[    (d-1) \left(\frac{\partial U}{\partial k_j}U^*A_j- G^jU^* \frac{\partial U}{\partial k_j}\right) -d G^jU^* A_j U +|k|\frac{\partial U}{\partial k_j} \frac{\partial A_j}{\partial r} - |k| G^j U^* \frac{\partial A_j}{\partial r}  U  \right]   \, . \]
 Note that $R^{d-1}\frac{d}{dR}  \left( R^{2-d} \int_{B_R} f \right)=(2-d) \int_{B_R}f-R \int_{\partial B_R} f$ for any integrable function $f$ and for a.e. $R>0$.
 Thus, dividing in \eqref{premonotonicity2} by $R^{d-1}$ and integrating on $[R_1,R_2]$ we have
 \begin{equation}
 \label{monformula}
 \frac{1}{R_1^{d-2}}\int_{B_{R_1}}  |\nabla U|^2+\int_{R_1<|k|<R_2} \frac{2}{|k|^{d-2}} \left|\frac{\partial U}{\partial r}\right|^2=\frac{1}{R_2^{d-2}}\int_{B_{R_2}}  |\nabla U|^2
 \end{equation}
 \[+\int_{R_1<|k|<R_2} V(k,U)+\int_{R_1}^{R_2} \frac{dR}{R^{d-1}} \int_{B_R} W(k,U) \, ,\]
where
\[V(k,U)= 2|k|^{2-d}\tr \left[  -\frac{\partial U}{\partial k_j}U^*A_j+ G^jU^* \left(\frac{\partial U}{\partial k_j}+A_j U\right)  + \frac{k_j}{|k|}\frac{\partial U}{\partial r}U^*A_j- \frac{k_j}{|k|}G^j U^* \frac{\partial U}{\partial r} \right] \]
and
\[W(k,U)= 2 \tr \left[    (d-1) \left(\frac{\partial U}{\partial k_j}U^*A_j- G^jU^* \frac{\partial U}{\partial k_j}\right) -d G^jU^* A_j U +|k|\frac{\partial U}{\partial k_j} \frac{\partial A_j}{\partial r} - |k| G^j U^* \frac{\partial A_j}{\partial r}  U  \right] \]
Notice that the first line in \eqref{monformula} is the usual
monotonicity identity for harmonic maps from $\mathbb{T}^*_d$  with
values into $\mathcal{U}(m)$. To estimate the extra terms in
\eqref{monformula}, observe that for fixed matrices $G^j$ we have
$|V(k,U)|\leq C|k|^{2-d} (1+|\nabla U|)$ where the constant depends
on the $C^0$ norm of the Berry connection matrices $A_j$. Similarly
we have $|W(k,U)|\leq C(1+|\nabla U|)$ where the constant depends
only on the $C^1$ norm of the Berry connection matrices. Thus, if
$l_0\geq 1$ is an integer and $R_1\geq 2^{-l_0}R_2$, $R_2\leq 1$ and
$0<\delta=R_2\leq 1$ we have (writing $\ \avint_{B}\ $ for $\frac{1}{|B|}\int_B$)

\[ \int_{R_1<|k|<R_2} \left| V(k,U) \right|\leq C \int_{R_1<|k|<R_2} |k|^{2-d}(1+|\nabla U|)\leq \]
\[ C R_2^2+C \sum_{l=0}^{l_0-1} \int_{2^{-l-1}R_2<|k|<2^{-l}R_2} |k|^{2-d}|\nabla U|\leq
CR_2^2+ C \sum_{l=0}^{l_0-1} 2^{-2l}R^2_2 \sqrt{  \avint_{B_{2^{-l}R_2}} |\nabla U|^2}\leq \]

\[ C R_2^2+ \delta \sup_{r\in [2^{-l_0}R_2,R_2]}\frac{1}{r^{d-2}} \int_{B_r} |\nabla U|^2  +C \delta^{-1} R_2^2 \, \leq C R_2 \left(1+ \sup_{r\in [2^{-l_0}R_2,R_2]}\frac{1}{r^{d-2}} \int_{B_r} |\nabla U|^2 \right) \, .
\]
On the other hand, a similar estimate with $\delta=\delta_l=2^{-l}R_2$ gives
\[ \int_{R_1}^{R_2} \frac{dR}{R^{d-1}} \int_{B_R} \left| W(k,U)\right| \leq  C R_2^2+ \int_{R_1}^{R_2} \frac{dR}{R^{d-1}} \int_{B_R} |\nabla U| \leq   C R_2^2+\]
\[ +C \sum_{l=0}^{l_0-1} \int_{2^{-l-1}R_2}^{2^{-l}R_2} \frac{dR}{R} \frac{1}{R^{d-2}} \int_{B_R} |\nabla U|\leq C \left[ R_2^2+ \left( \sum_{l=0}^{l_0-1} \delta_l  \right)  \sup_{r\in [2^{-l_0}R_2,R_2]}\frac{1}{r^{d-2}} \int_{B_r} |\nabla U|^2 \right. \]
\[  \left. + \sum_{l=0}^{l_0-1} \frac{1}{\delta_l} R^2\vert_{2^{-l-1}R_2}^{2^{-l}R_2} \right] \leq CR_2 \left( 1+ \sup_{r\in [2^{-l_0}R_2,R_2]}\frac{1}{r^{d-2}} \int_{B_r} |\nabla U|^2 \right)  \, . \]
Combining the two estimates above we obtain
\begin{equation}
\label{remestimate}
 \int_{R_1<|k|<R_2} \left|V(k,U) \right|+ \int_{R_1}^{R_2} \frac{dR}{R^{d-1}} \int_{B_R} \left| W(k,U)\right|\leq C R_2 \left(1+ \sup_{r\in [2^{-l_0}R_2,R_2]}\frac{1}{r^{d-2}} \int_{B_r} |\nabla U|^2 \right)
\end{equation}
Going back to \eqref{monformula}, if $R_2 \leq \bar{R}:= \min\{ 1,\frac12 C^{-1}\}$ is fixed, taking the supremum over $R_1\in [2^{-l_0}R_2, R_2]$  and estimating the right hand side using \eqref{remestimate} one easily obtains
\[ \sup_{r\in [2^{-l_0}R_2,R_2]}\frac{1}{r^{d-2}} \int_{B_r} |\nabla U|^2\leq 1+2 \frac{1}{R_2^{d-2}} \int_{B_{R_2}} |\nabla U|^2 \, .\]
Combining \eqref{monformula}, \eqref{remestimate} and the previous inequality we finally obtain \eqref{almostmon}.
Letting $R_1 \to 0$ in \eqref{almostmon} we see that $R^{2-d} \int_{B_R}  | \nabla U|^2dk$ is bounded and $ \int_{B_R}  | \partial_r U|^2|k|^{2-d}dk $ is finite and therefore vanishing as $R \to 0^+$. As a consequence, letting $R_1\to 0$ and $R_2 \to 0$ in \eqref{almostmon}, it is straightforward to see that $R^{2-d} \int_{B_R}  | \nabla U|^2dk$ has a limit as $R\to 0$.
\end{proof}

\begin{ciao}
\label{uniformlysmall}
A simple consequence of \eqref{almostmon} is that if $\lim_{R\to 0} R^{2-d} \int_{B_R(k_0)}  | \nabla U|^2dk=\varepsilon$ then there exists $R_0>0$ such that $\sup_{\bar{k} \in B_{R_0}(k_0)} \sup_{0<R\leq R_0}R^{2-d} \int_{B_R(\bar{k})}  | \nabla U|^2dk \leq 2\varepsilon$, \ie at sufficiently small scales the scaled energy is locally uniformly bounded by its limit at any point.
\end{ciao}

The second ingredient is a compactness theorem for the scaled maps $U_R(k)=U(k_0+Rk)$ which is similar to the compactness theorem for minimizing harmonic maps (see \cite{SU1} and \cite{S}, Chapter 2).

\begin{proposition}
\label{compactness}
Let $d\geq3$ and $U\in W^{1,2}(\mathbb{T}_d^*;\mathcal{U}(m))$ be a minimizer of the localization functional \eqref{MVU}. If we define $U_R(k)=U(k_0+Rk)$, $R>0$, then up to subsequences $U_R \to U_0$ strongly in $W^{1,2}_{\rm loc} (\mathbb{R}^d;\mathcal{U}(m))$. In addition $U_0$ is a locally minimizing harmonic map and it is degree-zero homogeneous.
\end{proposition}
To prove the previous result we need the following two auxiliary lemmas. The first is a simple consequence of local minimality. The second is a nonlinear  interpolation lemma due to Luckhaus which we state in our specific context (for the general case see e.g. \cite{LW}, Lemma 2.2.9; see also \cite{S}, Chapter 2 for a proof).

\begin{lemma}
\label{pertminimality}
Suppose $U$ and $U_R$ as in Proposition \ref{compactness}. Let $\rho \in (0,1)$ and let $\{ v_R \} \subset  W^{1,2}(B_\rho;\mathcal{U}(m))$ a bounded sequence such that $U_R=v_R$ on $\partial B_\rho$ for each $R>0$. Then
\[ \liminf_{R \to 0} \int_{B_\rho} |\nabla U_R|^2 \leq  \liminf_{R \to 0} \int_{B_\rho} |\nabla v_R|^2 \, . \]
\end{lemma}
\begin{proof}
Define $\tilde{v}_R(k)=v_R(R^{-1}(k-k_0))$ so that $\tilde{v}_R\in W^{1,2}(B_{\rho R}(k_0); \mathcal{U}(m))$. Since $\tilde{v}_R=U$ on $\partial B_{\rho R}(k_0)$ we can extend them as $U$ to the whole $\mathbb{T}^*_d$. The assumption on $v_R$ and Proposition \ref{monotonicity} clearly imply $\displaystyle{\int_{B_{\rho R}(k_0)} |\nabla U|^2+ |\nabla \tilde{v}_R|^2=\mathcal{O}(R^{d-2})}$ as $R\to 0$, hence formula \eqref{MVU} and simple calculations using Cauchy-Schwartz inequality give
\begin{equation}
\label{energydefect}
\tilde{F}_{MV}(\tilde{v}_R;\chi)-\tilde{F}_{MV}(U;\chi)=\int_{B_{\rho R}(k_0)}|\nabla \tilde{v}_R|^2-\int_{B_{\rho R}(k_0)}|\nabla U|^2+ \mathcal{O}(R^{d-1}) \, ,
\end{equation}
because $\tilde{v}_R$ and $U$ coincide outside $B_{\rho R}(k_0)$. Since $U$ is a minimizer of $\tilde{F}_{MV}$ the right hand side of \eqref{energydefect} is nonnegative, hence scaling back and taking the definition of $U_R$ and $\tilde{v}_R$ into account the conclusion follows as $R\to 0$.
\end{proof}
\begin{lemma}[Luckhaus]
\label{luckhaus}
Let $d\geq 3$, $m\geq 2$ and let $u,v \in W^{1,2}(S^{d-1};\mathcal{U}(m))$. Then, for each $\lambda \in (0,1)$ there is $w \in W^{1,2}(S^{d-1} \times (1-\lambda, 1); M_m(\mathbb{C}))$ such that $w|_{S^{d-1} \times \{1\}}=u$, $w|_{S^{d-1}\times \{1-\lambda\}}=v$,
\begin{equation}
\label{luck1}
\int_{S^{d-1} \times (1-\lambda,1)} |\nabla w|^2\leq C\lambda \int_{S^{d-1}}\left( |\nabla_T u|^2+ |\nabla_T v|^2 \right) +C\lambda^{-1}                 \int_{S^{d-1}} |u-v|^2
\end{equation}
and
\[ \mathrm{dist}^2(w(k), \mathcal{U}(m)) \leq C \lambda^{1-d} \left( \int_{S^{d-1}}\left( |\nabla_T u|^2+ |\nabla_T v|^2 \right)  \right)^{\frac12} \left(  \int_{S^{d-1}}|u-v|^2\right)^{\frac12} +\]
\begin{equation}
\label{luck2}
+C\lambda^{-d} \int_{S^{d-1}}|u-v|^2 
\end{equation}
for a.e. $k \in S^{d-1} \times (1-\lambda, 1)$. Here $\nabla_T$ is the gradient on $S^{d-1}$.
\end{lemma}

\begin{proof}[Proof of Proposition \ref{compactness}]
We essentially follow the proof of \cite{LW}, Lemma 2.2.13, with minor modifications.
 By Proposition \ref{monotonicity}, up to subsequences we may assume $U_R \rightharpoonup U_0$ in $W^{1,2}(B_1; \mathcal{U}(m))$ where $U_0$ is a degree-zero homogeneous map. Thus, it is enough to show strong convergence and minimality in some ball $B_\rho\subset B_1$ to get the same properties  on any $B_\rho\subset \mathbb{R}^d$ for any $\rho>0$, by scale invariance of $U_0$ and the existence of  the full limit of $R^{2-d} \int_{B_R}|\nabla U|^2$ as $R\to 0$.

Let $B_1 \subset \mathbb{R}^d$ and $\delta \in (0,1)$ a fixed number and let $\bar{w}\in W^{1,2}(B_1;\mathcal{U}(m))$  such that $\bar{w}\equiv U_0$ a.e. on $B_1\setminus B_{1-\delta}$.
By Fatou's lemma and Fubini's theorem, there is $\rho \in (1-\delta, 1)$ such that
\[ \lim_{R\to 0} \int_{\partial B_\rho} |U_R-U_0|^2 dH^{d-1}=0 \, , \,\, \int_{\partial B_\rho} \left( |\nabla U_R|^2+ |\nabla U_0|^2 \right)dH^{d-1}\leq C<\infty \, . \]
Applying Lemma \ref{luckhaus} to $\lambda=\lambda_R<\delta$, $u=U_R( \rho \, \cdot)$ and $v=\bar{w}(\rho \, \cdot) \equiv U_0(\rho \, \cdot)$ for a decreasing sequence of numbers $\lambda_R \to 0$, we conclude that there exists a sequence of maps $w_R \in W^{1,2}(B_\rho;M_m(\mathbb{C}))$ such that if we chose  e.g. $\displaystyle{\lambda_R=\left( \int_{\partial B_\rho} |U_R-U_0|^2 dH^{d-1}\right)^{1/2d}}<\delta$, then we have
\[ w_R(k)= \left\{
\begin{array}{ll}
\bar{w} \left( \frac{k}{1-\lambda_R} \right) \, , & |k| \leq \rho (1-\lambda_R) \, , \\
U_R(k) \, , & |k|=\rho \, ,
\end{array}
\right. \]
\begin{equation}
\label{noextenergy}
 \int_{B_\rho \setminus B_{\rho(1-\lambda_R)}} |\nabla w_R|^2 \leq C \left[  \lambda_R \int_{\partial B_\rho}\left( |\nabla_T U_R|^2+ |\nabla_T U_0|^2 \right) +\lambda_R^{-1}                 \int_{\partial B_\rho} |U_R-U_0|^2
\right] \overset{R\to 0}{\longrightarrow} 0 \, ,
\end{equation}
and $\mathrm{dist}(w_R, \mathcal{U}(m))\to 0$ uniformly on $B_\rho \setminus B_{\rho (1-\lambda_R)}$ as $R\to 0$. Define comparison maps $\{ v_R\} \subset W^{1,2}(B_\rho; \mathcal{U}(m))$ by
\begin{equation}
\label{defcomparison}
v_R(k)= \left\{
\begin{array}{ll}
\bar{w} \left( \frac{k}{1-\lambda_R} \right) \, , & |k| \leq \rho (1-\lambda_R) \, , \\
\Pi(w_R(k)) \, , & \rho (1-\lambda_R) \leq |k|\leq \rho \, ,
\end{array}
\right.
\end{equation}
where $\Pi: \mathcal{O} \to \mathcal{U}(m)$ is the nearest point projection.
Then, by minimality of $U_R$, Lemma \ref{pertminimality} and \eqref{noextenergy}-\eqref{defcomparison} we obtain
\[\int_{B_\rho} |\nabla U_0|^2 \leq \liminf_{R\to 0} \int_{B_\rho} |\nabla U_R|^2\leq \liminf_{R\to 0} \int_{B_\rho} |\nabla v_R|^2\]
\[ =\lim_{R\to 0} \left[ \int_{B_{\rho(1-\lambda_R)}} \left| \nabla \bar{w} \left( \frac{\cdot}{1-\lambda_R} \right) \right|^2+\int_{B_\rho\setminus B_{\rho(1-\lambda_R)}} \left| \nabla (\Pi \circ w_R) \right|^2  \right] \]
\[ \leq \lim_{R\to 0} \left[ (1-\lambda_R)^{d-2} \int_{B_\rho} \left| \nabla \bar{w}  \right|^2+C \, \mathrm{Lip}(\Pi)^2 \int_{B_\rho\setminus B_{\rho(1-\lambda_R)}} \left| \nabla  w_R \right|^2  \right] =\int_{B_\rho} |\nabla \bar{w}|^2 \, . \]
Since $\bar{w}$ is arbitrary, the previous inequality implies both minimality of $U_0$ and strong convergence $U_R\to U_0$ in $W^{1,2}(B_\rho; \mathcal{U}(m))$ as $R\to 0$ and concludes the proof.
\end{proof}
The final ingredient  is the following small-energy regularity result in the spirit of the fundamental $\varepsilon-$regularity theorem for harmonic maps proved in \cite{SU1}. The result is similar to \cite[Proposition 4.1]{Mos} but the stationarity condition as well as the argument of the proof there (the so-called "moving-frame"  trick) are different. Here we modify the elementary approach to regularity of \cite{CWY} for harmonic maps into spheres, by rewriting the right hand side of \eqref{EuleroLagrange}-\eqref{approxharmonic} in a suitable way and applying a standard estimate for the Laplace equation. Then, an iteration argument gives the decay of the BMO norm at small scales, whence continuity follows from the equivalence of Morrey-Campanato spaces and H\"{o}lder spaces in a suitable range of parameters.

\begin{proposition}
\label{epsregularity}
Let $d\geq 3$ and $U\in W^{1,2}(\mathbb{T}_d^*;\mathcal{U}(m))$ be a weak solution to the equations \eqref{EuleroLagrange}. Then there exist $\varepsilon>0$ and $\beta>0$,
 both independent of $U$ and $k_0\in \mathbb{T}^*_d$, such that if
$$
\sup_{\bar{k} \in B_{R_0}(k_0)} \sup_{0<R\leq R_0}R^{2-d} \int_{B_R(\bar{k})}  | \nabla U|^2dk \leq \varepsilon
$$ for some $R_0>0$, then $U \in C^{0,\beta}(B_{R_0/2}(k_0); \mathcal{U}(m))$.

\end{proposition}

Before going into the proof we quote two auxiliary results. First, recall that by definition for an open set $\Omega \subset \mathbb{R}^d$ a function $u\in L^1_{\rm loc} (\Omega)$ is in ${{BMO}}(\Omega)$ if
\begin{equation}
\label{defbmo}
\|u\|_{BMO(\Omega)}=\sup_{D} \avint_{D}|u-u_D|<\infty \, ,
\end{equation}
where $\displaystyle{u_D=\avint_{D} u}$ is the average of $u$ over $D$ and the supremum is taken over all balls $D\subset \Omega$. The first fact we need is a classical result of John and Nirenberg (see \cite{St}, Chapter 4).

\begin{lemma}
\label{johnnirenberg}
For any $1<p<\infty$, there exists a constant $C_p>0$ (which depends only on $p$ and $d$) such that if $u \in BMO(\Omega)$ then
\begin{equation}
\label{bmoequivalence}
\|u\|_{BMO(\Omega)}\leq \sup_{D} \left( \avint_{D}|u-u_D|^p \right)^{1/p} \leq C_p \| u\|_{BMO(\Omega)}<\infty \, ,
\end{equation}
where the supremum is taken over all balls $D\subset \Omega\subset \mathbb{R}^d$.
\end{lemma}

The second auxiliary result is a standard regularity property for solutions to the Laplace equation.

\begin{lemma}
\label{calderonzygmund}
Let $d\geq 3$ and $B_{\bar{R}} \subset \mathbb{R}^d$ be an open ball of radius $\bar{R}>0$.  Let $q \in (\frac{d}{d-1},2)$, $s=\frac{qd}{q+d}$. There exist $C>0$ depending only on $q$ such that if $F\in L^2(B_{\bar{R}};\mathbb{R}^d)$, $g \in L^2(B_{\bar{R}})$ and $u\in W^{1,2}_0(B_{\bar{R}})$ is a weak solution to $\Delta u=\dive F+g$, then
\begin{equation}
\label{CZ}
\| \nabla u \|_{L^q(B_{\bar{R}})} \leq C \left( \| F \|_{L^q(B_{\bar{R}})} +\| g \|_{L^s(B_{\bar{R}})}\right) \, .
\end{equation}
\end{lemma}

\begin{proof}[Proof of Proposition \ref{epsregularity}]
First note that if $U$ satisfies the condition

$$\displaystyle{\sup_{\bar{k} \in B_{R_0}(k_0)} \sup_{0<R\leq R_0}R^{2-d} \int_{B_R(\bar{k})}  |
\nabla U|^2dk \leq \varepsilon} \quad \hbox{ for some } R_0>0 \, ,
$$
 then on $B_{R_0}=B_{R_0}(k_0)$, by Cauchy-Schwartz and Poincaré inequality we have

\begin{equation}
\label{bmofinite}
\| U\|_{BMO(B_{R_0})} \leq \sup_{D_R \subset B_{R_0}} \left( \avint_{D_R}|U- \avint_{D_R}U|^2 \right)^{1/2} \leq C \sup_{D_R \subset B_{R_0}} \left( R^2 \avint_{D_R} | \nabla U|^2 \right)^{1/2}\leq C\varepsilon^{1/2} \, .
\end{equation}

Now, we aim to show that, for $\varepsilon$ sufficiently small, there is a quantitative decay of the BMO norm of $U$ at smaller and smaller scales.

Up to translation we may assume $k_0=0$. Let $\sigma \in (0,
\frac18]$ a fixed number to be specified later. For each
$\widehat{k} \in B_{R_0/2}$ and $t\in (0,R_0/2]$, let
$D_t=D_t(\widehat{k})\subset B_{R_0}$ be an open ball of radius $t$,
and for each $\bar{k}\in D_{\sigma t}=D_{\sigma t}(\widehat{k})$ let
$R\in (0,t)$ be such that $B_{\sigma R}(\bar{k}) \subset D_{\sigma
t}$. Clearly, $B_{R(\bar{k})} \subset D_t \subset B_{R_0}$, so that
if $\bar{R}\in (R/2,R)$ we still have the bound $\displaystyle{
\bar{R}^{2-d} \int_{B_{\bar{R}}(\bar{k})}  | \nabla U|^2dk \leq
\varepsilon}$.

On the other hand, given a constant matrix $T_0 \in M_m(\mathbb{C})$ with $|T_0|\leq \sqrt{m}$, e.g. $T_0=\avint_{B_R(\bar{k})} U$, we may choose $\bar{R} \in (R/2,R)$ so that
\begin{equation}
\label{averbound}
\int_{\partial B_{\bar{R}}(\bar{k})} |U-T_0| \leq 8 \int_{B_{R}(\bar{k})}|U-T_0|  \, .
\end{equation}
Since $U|_{\partial B_{\bar{R}}}\in W^{1/2,2}(\partial B_{\bar{R}};\mathcal{U}(m))$ there exists a harmonic extension $h \in W^{1,2}(B_{\bar{R}};M_m(\mathbb{C})$ so that h=$U|_{\partial B_{\bar{R}}}$ on $\partial B_{\bar{R}}$ and  $h \in C^\infty$ in the interior.
 Moreover, mean value formula, Jensen's inequality and \eqref{averbound} easily give
 \begin{equation}
 \label{pthbound}
 |\nabla h (k)|^p \leq C_p \bar{R}^{-p} \avint_{B_R(\bar{k})} |U-T_0|^p
 \end{equation}
for any $p \in (1,\infty)$ and any $k \in B_{\frac14 \bar{R}}(\bar{k})$.

On the other hand, as in the two dimensional case, since $U\in W^{1,2}(\mathbb{T}_d^*;\mathcal{U}(m))$ is a weak solution to \eqref{EuleroLagrange}, we have
 \begin{equation}
\label{approxharmonic2}
\Delta U= \sum_{j=1}^d \frac{\partial U}{\partial k_j} U^{-1} \frac{\partial U}{\partial k_j} +f \, ,
\end{equation}
for some $L^2$ function $f$ such that $U^{-1} f \in
L^2(\mathbb{T}^*_d; \mathfrak{u}(m))$. As in Proposition
\ref{approxhmreg}, if we set $B^j=\frac{1}{2}\left( U^*
\frac{\partial U}{\partial k_j} - \frac{\partial U^*}{\partial k_j}U
\right)$ then $ B^j \in L^{2}(\mathbb{T}_d^*; \mathfrak{u}(m))$ for
$j=1, \ldots, d$ and $U-h \in W^{1,2}_0(B_{\bar{R}}(\bar{k});
M_m(\mathbb{C}))$ is a weak solution to

\begin{equation}
\label{divform0}
\Delta (U-h)= \sum_{j=1}^d \frac{\partial U}{\partial k_j} B^j+f=\sum_{j=1}^d\frac{\partial}{\partial k_j} \left( (U-T_0) B^j \right)+f-(U-T_0) \sum_{j=1}^d \frac{\partial B_j}{\partial k_j} =
\end{equation}
\[ \hfill = \sum_{j=1}^d\frac{\partial}{\partial k_j} \left( (U-T_0) B^j \right)+f-(U-T_0) \frac{1}{2}(U^*f- f U^*) =\dive F +g  \, . \]

Since $|B(k)| \leq C |\nabla U(k)|$ and $|g(k)| \leq C |f(k)| \leq C (1+|\nabla U(k)|)$, for the right hand side of \eqref{divform0} we have the straightforward estimates

\begin{equation}
\label{auxcz1}
\| F\|_{L^q(B_{\bar{R}}(\bar{k}))} \leq \| \nabla U \|_{L^2(B_{\bar{R}}(\bar{k}))} \|  U -T_0\|_{L^{\frac{2q}{2-q}}(B_{\bar{R}}(\bar{k}))}\leq  \|  U -T_0\|_{L^{\frac{2q}{2-q}}(B_{\bar{R}}(\bar{k}))} \sqrt{\varepsilon} \bar{R}^{\frac{d-2}{2}} \, .
\end{equation}
and
\begin{equation}
\label{auxcz2}
\| g\|_{L^s(B_{\bar{R}}(\bar{k}))} \leq C  \bar{R}^{\frac{d}{s}}+ C \| \nabla U \|_{L^s(B_{\bar{R}}(\bar{k}))}\leq C  \bar{R}^{\frac{d-2}{2}} \left(\bar{R}^{2+d \left(\frac{1}{q}-\frac12 \right)} +  \sqrt{\varepsilon} \bar{R}^{d\frac{2-s}{2}} \right)\, .
\end{equation}
Applying Lemma \ref{calderonzygmund} to \eqref{divform0} and taking \eqref{auxcz1} and \eqref{auxcz2} into account we obtain
\begin{equation}
\label{gradlqest0}
\int_{B_{\bar{R}}(\bar{k})} |\nabla (U-h)|^q \leq  C_q \varepsilon^{\frac{q}{2}} \bar{R}^{q\frac{d-2}{2}} \|  U -T_0\|^q_{L^{\frac{2q}{2-q}}(B_{\bar{R}}(\bar{k}))}+C_q \bar{R}^{q\frac{d-2}{2}} \left(\bar{R}^{2q+d \frac{2-q}{2} } + \varepsilon^{\frac{q}{2}} \bar{R}^{qd\frac{2-s}{2}} \right) \, ,
\end{equation}
i.e.
\begin{equation}
\label{gradlqest}
\avint_{B_{\bar{R}}(\bar{k})} |\nabla (U-h)|^q \leq C_q  \bar{R}^{-q} \left[ \varepsilon^{\frac{q}{2}}\left( \avint_{B_{\bar{R}}(\bar{k})} |  U -T_0|^{\frac{2q}{2-q}}  \right)^{\frac{2-q}{2}}  +C_q \bar{R}^{2q}  + C_q\varepsilon^{\frac{q}{2}} \bar{R}^{qd\frac{2-s}{2}+q\frac{d}{2}-d} \right] \, .
\end{equation}

Now we choose $p=q^*=\frac{dq}{d-q}>q$  and $\underline{R}=\sigma R<\bar{R}<R$. Using Sobolev inequality, \eqref{gradlqest} and \eqref{pthbound}  we estimate
\[ \avint_{B_{\sigma R}(\bar{k})}|U-h(\bar{k})|^p \leq \frac{C}{\underline{R}^d} \int_{B_{\bar{R}}(\bar{k})} |U-h|^p+\frac{C}{\underline{R}^d}\int_{B_{\underline{R}}(\bar{k})} |h-h(\bar{k})|^p   \]
\[ \leq C\frac{\bar{R}^{d+p}}{\underline{R}^d} \left(\avint_{B_{\bar{R}}(\bar{k})} |\nabla(U-h)|^q \right)^{\frac{p}{q}}+C\underline{R}^p \sup_{B_{\bar{R}/4}} |\nabla h|^p \]
\[ \leq C \sigma^{-d} \left(    \varepsilon^{\frac{q}{2}}\left( \avint_{B_{\bar{R}}(\bar{k})} |  U -T_0|^{\frac{2q}{2-q}}  \right)^{\frac{2-q}{2}}  +C_q \bar{R}^{2q}  + C_q\varepsilon^{\frac{q}{2}} \bar{R}^{qd\frac{2-s}{2}+q\frac{d}{2}-d} \right)^{\frac{p}{q}}+C  \sigma^p \avint_{B_R(\bar{k})} |U-T_0|^p \]
\[ \leq C \sigma^{-d} \left(    \varepsilon^{\frac{p}{2}}\left( \avint_{B_R(\bar{k})} |  U -T_0|^{\frac{2q}{2-q}}  \right)^{p\frac{2-q}{2q}}  + R^{2p}  + \varepsilon^{\frac{p}{2}} R^{pd\frac{2-s}{2}+p\frac{d}{2}-p-d} \right)+C  \sigma^p \avint_{B_R(\bar{k})} |U-T_0|^p \]
Since $B_R(\bar{k}) \subset D_t$, if we choose $T_0=\avint_{B_R(\bar{k})} U$ the John-Nirenberg inequality \eqref{bmoequivalence} yields
\begin{equation}
\label{bmo1}
\left( \avint_{B_{\sigma R}(\bar{k})}|U-h(\bar{k})|^p \right)^{1/p} \leq \left(  \sigma^{-d/p} \varepsilon^{1/2} +\sigma \right) C_q \| U\|_{BMO(D_t)}+C_q \sigma^{-d/p}  \left( t^{2p}  + \varepsilon^{\frac{p}{2}} t^{pd\frac{2-s}{2}+p\frac{d}{2}-p-d} \right) \, .
\end{equation}
On the other hand, by H\"{o}lder inequality and \eqref{bmo1} we get
\begin{equation}
\label{bmo2}
 \avint_{B_{\sigma R}(\bar{k})}|U-\avint_{B_{\sigma R}(\bar{k})} U|
 \leq \left( \avint_{B_{\sigma R}(\bar{k})}|U-h(\bar{k})|^2 \right)^{1/2}
\leq \left( \avint_{B_{\sigma R}(\bar{k})}|U-h(\bar{k})|^p \right)^{1/p}
\end{equation}
\[
  \leq \left(  \sigma^{-d/p} \varepsilon^{1/2} +\sigma \right) C_q \| U\|_{BMO(D_t)}+C_q \sigma^{-d/p}  \left( t^{2p}  + \varepsilon^{\frac{p}{2}} t^{pd\frac{2-s}{2}+p\frac{d}{2}-p-d} \right) \, ,
\]
hence, taking the supremum over $B_{\sigma R}(\bar{k}) \subset D_{\sigma t}$ we  obtain
\begin{equation}
\label{bmo3}
\|U\|_{BMO(D_{\sigma t})} \leq  \left(  \sigma^{-d/p} \varepsilon^{1/2} +\sigma \right) C_q \| U\|_{BMO(D_t)}+C_q \sigma^{-d/p}  \left( t^{2p}  + \varepsilon^{\frac{p}{2}} t^{pd\frac{2-s}{2}+p\frac{d}{2}-p-d} \right) \, .
\end{equation}
Since  $p=p(q)$, $s=s(q)$, the exponent $\alpha:=pd\frac{2-s}{2}+p\frac{d}{2}-p-d \to 1 $ as $q \searrow \frac{d}{d-1}$ so we can fix $q$ small such that $\alpha \in (0,2)$. If we choose $\sigma \in (0,\frac18]$ and $\varepsilon>0$ so small that $C_q ( \sigma^{-d/p} \varepsilon^{1/2} +\sigma) <\frac12$ then
\begin{equation}
\label{bmo4}
\|U\|_{BMO(D_{\sigma t})} \leq \frac12 \| U\|_{BMO(D_t)}+C t^\alpha \, , \qquad \forall t\in (0, R_0/2 ] \, ,
\end{equation}
where $C>0$ and $\alpha\in (0,2)$ are independent of $U$,
$\widehat{k}$ and  $t$. Thus, from \eqref{bmo4} an elementary
iteration argument on $v(t)= \| U\|_{BMO(D_t)}$ and the
John-Nirenberg inequality \eqref{bmoequivalence} give
\begin{equation}
\label{holdercontinuity} \avint_{D_t(\widehat{k})} |U-
\avint_{D_t(\widehat{k})} U |^2 \leq C t^{2\beta} \, , \qquad
\forall t \in (0, R_0/2] \, , \quad \forall \widehat{k} \in
B_{R_0/2} \, ,
\end{equation}
for some $\beta=\beta(\alpha)>0$.
Finally, from \cite{Ca} we conclude that $U \in C^{0,\beta}(B_{R_0/2};\mathcal{U}(m))$ and the proof is complete.
\end{proof}

Combining the previous propositions and the Liouville type theorem proved in the Appendix we have the main result of this section.

\begin{theorem}
\label{3dcontinuity}
Let $d=3$ and $U\in W^{1,2}(\mathbb{T}_d^*;\mathcal{U}(m))$ be a minimizer of the localization functional \eqref{MVU}. If $2\leq m \leq3$ then $U\in C^0(\mathbb{T}_d^*;\mathcal{U}(m))$.
\end{theorem}
\begin{proof}
Fix $k_0 \in \mathbb{T}^*_3$ and for each $R=R_n\searrow 0$ we define $U_R(k)=U(k_0+Rk)$. According to Proposition \ref{monotonicity} such a sequence is bounded in $W^{1,2}_{\rm loc} (\mathbb{R}^3;\mathcal{U}(m))$ and, up to subsequences,
it converges weakly to a degree-zero homogeneous map $U_0 \in W^{1,2}_{\rm loc} (\mathbb{R}^3;\mathcal{U}(m))$. According to Proposition \ref{compactness}, such convergence is strong and the limiting map      $U_0$ is a degree-zero homogeneous local minimizer of the Dirichlet integral in $W^{1,2}_{\rm loc} (\mathbb{R}^3;\mathcal{U}(m))$. According to Corollary \ref{liouville}, when $2\leq m \leq3$ we have $U_0(k)\equiv const$, therefore $\frac{1}{R_n} \int_{B_{R_n}(k_0)} |\nabla U|^2= \int_{B_1} |\nabla U_{R_n}|^2\to 0$ as $n \to \infty$, hence  it can be made arbitrarily small at sufficiently small scale. Thus, in view of Proposition \ref{epsregularity} continuity around $k_0$ follows, and $U\in C^0(\mathbb{T}^*_3;\mathcal{U}(m))$ since $k_0$ was arbitrary.
\end{proof}

\section{Analytic regularity}
In this section we first prove analytic regularity for continuous
weak solutions to the Euler Lagrange equations
\eqref{EuleroLagrange}, whenever $\chi$ is a real-analytic Bloch
frame (\ie $A_j \in C^{\omega}(\Base, \mathfrak{u}(m))$). Then,
combining this property with the continuity results for the
minimizers of the localization functional \eqref{MVU}, we prove
analyticity for any minimizer of the functionals \eqref{mvfimbande}
and \eqref{MVU}.

We start with the following auxiliary result.

\begin{proposition}
\label{strongsolution} Assume $d \geq 2$. Let $\chi$ be a
real-analytic Bloch frame and $U\in
W^{1,2}(\mathbb{T}_d^*;\mathcal{U}(m))$ be a weak solution to
equation \eqref{EuleroLagrange}. If $U\in
C^0(\mathbb{T}_d^*;\mathcal{U}(m))$ then $U\in
W^{1,4}(\mathbb{T}_d^*;\mathcal{U}(m))$ and $U\in
W^{2,2}(\mathbb{T}_d^*;\mathcal{U}(m))$.
\end{proposition}
\begin{proof}
We rewrite the system \eqref{EuleroLagrange} in the form
\begin{equation}
\label{nonlinearsist}
\Delta U= \mathcal{F}(A(k),U,\nabla U)=F(k,U,\nabla U) \, ,
\end{equation}
where $\mathcal{F}:\left(M_m(\mathbb{C}) \right)^d \times
M_m(\mathbb{C}) \times \left( M_m(\mathbb{C})\right)^d\to
M_m(\mathbb{C})$ is a real-analytic (polynomial) map and the
matrices $A(k)=(A_1(k), \ldots, A_d(k))$, \ie the entries of the
Berry connection given in \eqref{berry}, are real-analytic on  the
torus $\mathbb{T}^*_d$. It is easy to see that, since $A$ is smooth
and $U$ takes values into $\mathcal{U}(m)$, by construction the
function $F(k,s,p)$ on the range of $U$ satisfies the structural
assumptions
\begin{equation}
\label{structG}
|F(k,s,p)|+|\nabla_s F(k,s,p)|\leq c_0(1+|p|^2) \, , \qquad |\nabla_k F(k,s,p)|+|\nabla_p F(k,s,p)| \leq c_1(1+|p|) \,
\end{equation}
on $\mathbb{T}^*_d \times \mathcal{U}(m) \times \left( M_m(\mathbb{C})\right)^d$.

As a consequence of \cite{J}, Lemma 8.5.1 and Lemma 8.5.3, any continuous weak solution $U$ of \eqref{nonlinearsist} is locally in $W^{2,2}\cap W^{1,4}$ and the conclusion follows taking a finite cover of the torus.
\end{proof}

Combining the previous result with the standard regularity theory
for linear elliptic equations and the fundamental analyticity
results for nonlinear elliptic systems (see \eg \cite{Morrey},
Chapter VI), we obtain full regularity. The proof is standard, so we
just sketch it for the reader's convenience.

\begin{proposition}
\label{analitycity} Let $d=2$ or $d=3$  and let $U\in
W^{1,2}(\mathbb{T}_d^*;\mathcal{U}(m))$ be a weak solution to
equation \eqref{EuleroLagrange}. If $\chi$ is real-analytic and
$U\in C^0(\mathbb{T}_d^*;\mathcal{U}(m))$, then $U$ is
real-analytic.
\end{proposition}

\begin{proof}
According to Proposition \ref{strongsolution} we know that $U\in
W^{1,4}(\mathbb{T}_d^*;\mathcal{U}(m))$ and $U\in
W^{2,2}(\mathbb{T}_d^*;\mathcal{U}(m))$. If $d=2$ then Sobolev
embedding yields $\nabla U \in L^p$ for any $p<\infty$. Similarly,
in case $d=3$ we have $\nabla U \in L^6$, hence \eqref{structG}
implies $G(k,U,\nabla U)\in L^3$. Thus, linear elliptic regularity
for \eqref{nonlinearsist} gives $U\in W^{2,3}$ and in turn $\nabla U
\in L^p$ for any $p<\infty$ again by Sobolev embedding (note that
the same property holds for any $d\geq 4$, compare \cite[Lemma
8.5.4]{J}). Clearly, if $\nabla U \in L^p$ for any $p<\infty$ the
same is true for $G(k,U,\nabla U)$ because of \eqref{structG}, hence
linear elliptic regularity for \eqref{nonlinearsist} gives $U\in
W^{2,p}$ for any $p<\infty$, which in turn yields $U\in
C^{1,\alpha}$ for any $\alpha \in (0,1)$ by Sobolev-Morrey
embedding. Going back to \eqref{EuleroLagrange} a standard bootstrap
argument in the H\"{o}lder spaces $C^{l,\alpha}$, $l\geq 1$, yields
by induction $U\in C^{l,\alpha} \Rightarrow \Delta U \in
C^{l-1,\alpha} \Rightarrow U \in C^{l+1,\alpha}$, so that $U\in
C^\infty(\mathbb{T}^*_d;\mathcal{U}(m))$. Finally, since the
coefficients $A_j$ in \eqref{EuleroLagrange} are analytic, by the
results in \cite{Morrey}, Chapter VI,  any smooth solution is
real-analytic.
\end{proof}

\renewcommand{\theenumi}{\roman{enumi}}
The main result of the section is the following.
\begin{theorem}\label{Cor real analiticity} Let $1 \leq d \leq 2$ and $m \geq 1$,
 or $d=3$ and $1 \leq m \leq3$, or $d\geq 4$ and $m=1$. Let $\chi$
be a  real-analytic Bloch frame and $\tilde F_{MV}(\cdot)$ and
$\tilde F_{MV}(\cdot\,;\, \chi)$ the functionals defined by
(\ref{mvfimbande}) and (\ref{MVU}) respectively. Then:

\begin{enumerate}
    \item any minimizer $U \in W^{1,2}(\Base; \Un(m))$ of $\tilde
    F_{MV}(\cdot\,;\,
    \chi)$ is real-analytic;
    \item any minimizer $\ph = \set{\ph_1, \ldots, \ph_m} \subset \tilde
    \W$ of $\tilde F_{MV}(\cdot)$ is a real-analytic map from $\R^d$
    to $(\Hf)^m$.
\end{enumerate}
\end{theorem}

\begin{proof}
\textrm{(i)} Since any minimizer $U$ is a weak solution to
(\ref{EuleroLagrange}) and $\chi$ is real-analytic, the conclusion
follows from Proposition \ref{regsingleband} if $d=1$ or $m=1$, and
from Corollary \ref{2dcontinuity}, Theorem \ref{3dcontinuity} and
Proposition \ref{analitycity} in the other cases.

\noindent \textrm{(ii)} Let $w = \UZ^{-1}\ph$ and $\ph = \chi \cdot
U$. Then $F_{MV}(w) = \tilde{F}_{MV}(\ph)= \tilde{F}_{MV}(U; \chi)$,
$U$ is a minimizer of the latter functional in view of (\ref{Inf
equality}), and the conclusion follows from part \textrm{(i)} above.
\end{proof}


\section{Exponential localization of maximally localized Wannier functions}

The main result of the paper is the following, and provides an
affirmative answer to problems $\mathrm{(MV_1)}$ and
$\mathrm{(MV_2)}$.

\begin{theorem}
\label{maintheorem} Let $\sigma_*$ be a family of $m$ Bloch bands
for the operator (\ref{Hamiltonian}) satisfying the gap condition
(\ref{Gap condition}), and let  $\{ P_{*}(k)\}_{k \in \R^d}$ be the
corresponding family of spectral projectors. Assume $d \leq 2$ and
$m \geq 1$, or $d\geq 1$ and $m=1$, or $d=3$ and $1\leq m \leq3$.  Then there exist
composite Wannier functions $\set{w_1, \ldots, w_m} \subset \W$
which minimize the localization functional \eqref{mvwmbande} under
the constraint that the corresponding quasi-Bloch functions are an
orthonormal basis for $\Ran \, P_*(k)$ for each $k \in Y^*$. In
addition, for any system of maximally localized composite Wannier
function $w = \set{w_1, \ldots, w_m}$ there exists $\beta >0$ such
that $e^{\beta |x|} w_a$ is in $L^2(\R^d)$ for every $a \in \set{1,
\ldots, m}$, \ie the composite Wannier function $w_a$ is
exponentially localized.
\end{theorem}

Conjecturally, we expect that the parameter $\beta$ appearing in the
latter claim does not depend on the minimizer $w$, and that the
claim holds true for any $\beta < \alpha$, where $\alpha$ is
appearing in (\ref{Omega}).

\begin{proof} In view of Theorem \ref{Theorem triviality}, there
exists a Bloch frame $\chi$ which is  real-analytic. Therefore,
problem $\mathrm{(MV_1)}$ is equivalent to showing
 that the r.h.s.\ of
(\ref{Inf equality}) is attained, which is proved in Theorem
\ref{existence}.

Let $w = \set{w_1, \ldots w_m} \subset \W$ be any minimizer of
$F_{MV}$. Then $\ph_a := \UZ w_a$  defines a minimizer $\set{\ph_1,
\ldots, \ph_m} \subset \tilde \W$ of $\tilde F_{MV}$ among the Bloch
frames. With respect to the frame $\chi$, one has $\ph = \chi
\,\cdot \, U$, where $U \in W^{1,2}(\Base; \Un(m))$ is given by
$U_{b,a}(k) = \inner{\chi_b(k)}{\ph_a(k)}$. Clearly, $U$ is a
minimizer of $\tilde F_{MV}(\cdot\,;\, \chi)$. By Theorem \ref{Cor
real analiticity}, $U \in W^{1,2}(\Base; \Un(m))$ is actually
real-analytic.

The function $U$ defines a real-analytic $\Gamma^*$-periodic
function $\tilde U: \R^d \to \Un(m)$. By unique continuation,
$\tilde U$ extends to a holomorphic function $\tilde U^{\C}$ from
$\Omega_{\beta_1}$ to $\mathrm{GL}(m, \C)$, which is
$\Gamma^*$-periodic in the real part of its argument. Analogously, arguing as in the
proof of Proposition \ref{Prop P properties}, $\chi$ admits a
holomorphic extension $\chi^{\C} \in \Hi_{\tau, \beta_2}^{\C}$ for
some $\beta_2 > 0$. Therefore $\ph = \chi \cdot U$ is a
real-analytic Bloch frame which admits a holomorphic extension
$\ph^{\C} = \chi^{\C} \cdot \tilde{U}^{\C}$, which is in $\Hi_{\tau,
\beta_0}^{\C}$ for $\beta_0 = \min\{\beta_1, \beta_2, \alpha \}$.
Moreover, for any $\beta < \beta_0$ there exists $C$ such that
$$
\int_{Y*} \norm{\ph^{\C}(k + i h)}^2 \,dk <C
$$
for every $h$ such that $|h_j| \leq \beta$ for $j \in \set{1,
\ldots,d}$. By Proposition \ref{Lemma Paley-Wiener} one has that
$w_a := \UZ^{-1}\, \ph_a$ satisfies
$$
\int e^{2 \beta |x|}\,|w_a(x)|^2 \, dx < + \infty.
$$
\end{proof}

\newpage



\appendix
\section{Harmonic maps into $\mathcal{U}(m)$}

We consider, for $\Omega^\prime = \mathbb{R}^3 \setminus \{0 \}$ and
$U \in W^{1,2}_{\rm loc} (\Omega^\prime; \, \mathcal{U}(m))$ with
$m\geq 2$, the energy functional
\begin{equation}
\label{Dirichletenergy} E(U;\Omega)= \int_{\Omega} \frac12
\sum_{j=1}^3 \tr \left( \frac{\partial U^*}{\partial
k_j}\frac{\partial U}{\partial k_j}\right) dk \, , \qquad \Omega
\subset\subset \Omega^\prime .
\end{equation}
We assume that $U$ is a local minimizer of \eqref{Dirichletenergy}
in $\Omega^\prime$, \ie that $E(U;\Omega)\leq E(W;\Omega)$ for any
$\Omega \subset \subset \Omega^\prime$ and for any $W \in
W^{1,2}_{\rm loc}(\Omega^\prime; \, \mathcal{U}(m))$ such that
$\supp(U-W) \subset \subset \Omega$. Clearly, if $\Psi \in
C^\infty_0(\Omega; \, \mathfrak{u}(m))$ and $\varepsilon\in
\mathbb{R}$, then $U_\varepsilon(k)=U(k) \exp{\varepsilon \Psi(k)}$
is an admissible variation of $U$, hence local minimality gives
\begin{equation}
\label{locminimality} \left. \frac{d}{d\varepsilon}
\right\vert_{\varepsilon=0} \, E(U_\varepsilon;\Omega) =0 \, ,
\qquad  \left. \frac{d^2}{d\varepsilon^2}
\right\vert_{\varepsilon=0} \, E(U_\varepsilon;\Omega) \geq 0 \, .
\end{equation}
Since the tangential variation $\Psi$ can be chosen arbitrarily, the
first condition in \eqref{locminimality} easily implies that $U$ is
a weakly harmonic map, \ie$U$ is a weak solution to
\begin{equation}
\label{harmonicmap} -\Delta U+ \sum_{j=1}^3 \frac{\partial
U}{\partial k_j} U^{-1} \frac{\partial U}{\partial k_j}=0 \, .
\end{equation}

We aim to prove that, when $m\geq2$, any local minimizer $U \in
W^{1,2}_{\rm loc}(\Omega^\prime; \,\mathcal{U}(m))$ which is
degree-zero homogeneous (a minimizing tangent map), \ie
$$
U(k)=\omega \left( \frac{k}{|k|}\right) \mbox{ for some } \omega \in
C^\infty (S^2; \, \mathcal{U}(m)),
$$ is constant.  The argument we are going to use combines a stability
inequality derived from \eqref{locminimality} (see inequality
\eqref{Stability ineq} below) and a nontrivial quantization property
for the energy of every harmonic map $\omega \in C^{\infty}(S^2; \,
\mathcal{U}(m))$  \cite[ Corollary 8]{V}. \newline 
Actually, in view of Lemma \ref{Lemma SU} below, when $m=2$ this
constancy property is known and it follows from \cite[Proposition
1]{SU2}, since $\mathcal{SU}(2) \equiv S^3(\sqrt{2})$.  However,
when the target is a sphere, the constants in the stability
inequalities are uniformly bounded as the dimension of the sphere
increases (see \cite[formula $(*)$]{SU2}), which will not be the
case in the problem we are dealing with. Here we prove, by a
different technique, the constancy property in the case $m \leq 3$.   In our opinion, if
\eqref{Stability ineq}  cannot be improved, then it seems difficult
to prove the Liouville property for any $m\geq 4$ using the
so-called Bochner method as in the sphere-valued case (see
\cite{SU2} and \cite{LW2}; see also \cite[Chapter 5]{X} and
references therein).


As far as our specific problem is concerned, we can assume that the
target is indeed the special unitary group $\mathcal{SU}(m)$ in view of the
following auxiliary result.

\begin{lemma}\label{Lemma SU} Let $d=3$ and $m \geq 2$. Let $U \in W^{1,2}_{\rm
loc}(\Omega^{\prime}; \, \mathcal{U}(m))$ be a degree-zero
homogeneous weakly harmonic map.   Then $\det \, U \equiv \alpha \in
\mathcal{U}(1)$  and $U \in W^{1,2}_{\rm loc}(\Omega^{\prime};\,
\mathcal{SU}(m))$ up to multiplication by a constant unitary matrix
$U_0$. As a consequence, $\tr (U^*
\partial_j U) \equiv 0$ for each $1\leq j\leq 3$.
\end{lemma}


\begin{proof}
Clearly, $\det U\in W^{1,2}_{\rm loc}(\Omega^{\prime};\,
\mathcal{U}(1))$ and it is degree-zero homogeneous. According to
\cite{BZ}, if $B$ is the unit ball, there exists $g\in W^{1,2}(B)$
such that $\det U= \expo{ig}$ a.e.\ in $B$. By slicing, $g$ is
$W^{1,2}_{\rm loc} $ on a.e.\ ray from the origin, so it is
continuous along a.e.\ ray. Since $\expo{ig}= \det U$ is constant
along the rays, we conclude that $g$ is also constant along the
rays, \ie $g$ is degree-zero homogeneous.

Let us set $\hat{U}(k)=\expo{\frac{i}{m}g(k)} \, \Id$ and
$W(k)=\hat{U}(k)^* U(k)$. By construction $\det \hat{U}\equiv \det
U$, so that $W=\hat{U}^* U \in W^{1,2}(B;\, \mathcal{SU}(m))$. Thus,
in order to prove the lemma it is clearly enough to show that $g$
(and in turn $\hat{U}$) is constant in $B$, because the conclusion
follows in the whole $\mathbb{R}^3$  since $U$ is
degree-zero homogeneous. Notice that if $\eta \in C^\infty_0(B)$,
$\varepsilon \in \mathbb{R}$ and $g_\varepsilon=g+\varepsilon \eta$
then
\begin{equation}
\label{energydec} E(\expo{\frac{i}{m}g_\varepsilon} \Id \, W;
B)=\int_B\frac{1}{2m} |\nabla g_\varepsilon|^2+ \int_B \frac12
\sum_{j=1}^3 \tr \left( \frac{\partial W^*}{\partial
k_j}\frac{\partial W}{\partial k_j}\right) dk
=E(\expo{\frac{i}{m}g_\varepsilon} \Id;B)+E(W;B) \, .
\end{equation}
Since $U_\varepsilon=\expo{\frac{i}{m}g_\varepsilon} \Id \, W$ is an
admissible variation for $U$, differentiating \eqref{energydec} we
readily see that the function $g$ is weakly harmonic in $B$, hence
$g$ is continuous (real-analytic). Since $g$ is also degree-zero
homogeneous we conclude that $g$ is constant in $B$ as claimed. 
As a consequence, $\hat{U}$ is constant, $\det U$ is also constant
(both in $B$ and in $\mathbb{R}^3$, both functions being degree-zero
homogeneous) and $U \in W^{1,2}_{\rm loc}(\Omega^{\prime}; \,
\mathcal{SU}(m))$ up to multiplying by a constant unitary matrix
$U_0=\hat{U}^*$. Finally, since $U^*\partial_jU \in
\mathfrak{su}(m)$ we also have $\tr \(U^*
\partial_j U\) \equiv 0$ for each $1\leq j\leq 3$.
\end{proof}


\subsection{The stability inequality}
Throughout this section we assume that $U$ is a smooth harmonic map
and we denote by $V$ the variational vector field, \ie $ V(k) =
\left. \frac{d}{d\varepsilon} \right\vert_{\varepsilon=0} \,
U_\epsi(k) = U_{\epsi}(k) \Psi(k)$ associated to the deformation
$U_\epsi (k)=U(k) \exp^{\epsi \Psi(k)}$. We regard
$\mathcal{SU}(m)\subset M_m(\mathbb{C})$ as a Riemannian manifold
with the metric induced by the embedding \ in $M_m(\C)$, the latter
being equipped with the Hilbert-Schmidt inner
product\!\footnote{\ Notice that this metric differs by a constant
from the metric on $\mathcal{SU}(m)$ induced by
the Killing form on $\mathfrak{su}(m)$. The difference is, for our purposes, immaterial.}
\ $\inner{A}{B} = \re \tr(A^*B)$. The second variation formula for
the energy \cite[Chapter 1]{LW} yields
\begin{equation}\label{Second variation}
\left. \frac{d^2}{d\varepsilon^2} \right\vert_{\varepsilon=0} \,
E(U_\varepsilon;\Omega) = \sum_j \int_{\Omega}
\normHS{\widetilde\nabla_{e_j} V}^2  - R(dU(e_j), V, dU(e_j), V)
\end{equation}
where $\{ e_j \}$ is an orthonormal basis of $\R^3 \cong
T\Omega^{\prime}$, $\widetilde \nabla$ is the pull-back via $U$ of
the Levi-Civita connection on $T\mathcal{SU}(m)$, $R$ is the curvature
$(4,0)$-tensor on $T\mathcal{SU}(m)$, and $\normHS{A}=
(\tr(A^*A))^{1/2}$ is the Hilbert-Schmidt norm in $M_m(\C)$.

In order to obtain a convenient stability inequality, we focus on a
variational vector field $V$ which is obtained by  orthogonal
projection onto $T\mathcal{SU}(m)$ of a constant vector field on
$M_m(\C)$ and we average over such constant vector fields the
corresponding inequality given by \eqref{Second variation}. The idea
is not new, it is explicitly used in \cite{SU2} when the target is a
sphere (averaging over the conformal vector fields) and more
generally in \cite{HW} for homogeneous manifolds and \eg in
\cite{We} for general targets. Here we follow a very concrete and
elementary approach when the target is $\mathcal{SU}(m) \subset
M_m(\mathbb{C})$ and we obtain a stability inequality with an
explicit constant (see inequality \eqref{Stability ineq} below).

\begin{ciao}
\label{howardwei} Following \cite{HW} the same inequality (with
exactly the same constant) could be deduced in a more abstract way,
regarding $\mathcal{SU}(m)$ as a homogeneous (group) manifold. More
precisely, one can regard $\mathcal{SU}(m)$ as minimal submanifolds
in the sphere $S^{m^2-1}(\sqrt{m})$ through the standard minimal
immersion in the first nontrivial eigenspace of its Laplace-Beltrami
operator, getting a stability inequality with an explicit constant
expressed in terms of the first nonzero eigenvalue of the Laplacian
(see e.g. \cite{X}, pages 137-138 and equation (5.23) or \cite{HW},
pages 328-329 and Proposition 5.2). The eigenvalue as well as the
dimension of the eigenspace are known in the literature, in terms of
the representation theory of the Lie algebra $\mathfrak{su}(m)$, so
the constant can be explicitly computed.
\end{ciao}

To implement this idea, we consider the orthogonal projection $P_U:\
M_m(\C) \rightarrow T_U\mathcal{SU}(m)$ defined by
$$
P_U(\phi) = \half \, U \left( U^* \phi - \phi^* U - \frac{1}{m}
\tr(U^* \phi - \phi^* U) \, \Id \right)
$$
Clearly, for $U = \Id$, the formula above reduces to the orthogonal
projection from $M_m(\C)$ onto $\mathfrak{su}(m)$. For a fixed $\phi
\in M_m(\C)$, we define the tangent vector field $\transpose{\tilde{\phi}}$ along $U$
by setting  $$\transpose{\tilde{\phi}}(k) =
P_{U(k)}(\phi).$$ For $\eta \in C_0^{\infty}(\Omega^{\prime}, \R)$,
we define $ V^{\phi, \eta}(k) := \eta(k)
\transpose{\tilde{\phi}}(k)$, as a section of the pull-back bundle $U^* T\mathcal{SU}(m)$, corresponding to the admissible variation $U^{\phi,
\eta}_\epsi(k)=U(k) \exp(\epsi \eta(k) U^*(k) \transpose{\tilde{\phi}}(k))$. The advantage of this choice is that the
second variation of the energy, when averaged with respect to $\phi$
varying in an orthonormal basis (ONB) of $TM_m(\C) \cong M_m(\C)$,
decouples as the sum of two simpler terms, as shown in the following
lemma.

\begin{lemma}\label{Lemma Energy Decomposition}
Let $V^{\phi, \eta}$ be defined as above, and $U_{\epsi}^{\phi,
\eta}$ be the corresponding variation. Then
\begin{eqnarray} \label{Energy var decomposition}
\nonumber \sum_{\phi \in \mathrm{ONB}} \left.
\frac{d^2}{d\varepsilon^2} \right\vert_{\varepsilon=0} \, E(U^{\phi,
\eta}_\epsi;\Omega) &=& (m^2-1) \int_{\Omega} |\nabla \eta|^2 \\
&+& \int_{\Omega} \eta^2 \sum_j \sum_{\phi \in \mathrm{ONB}} \left\{
\normHS{\widetilde \nabla_{e_j} \transpose{\tilde{\phi}}}^2 - R(dU(e_j),
\transpose{\tilde{\phi}}, dU(e_j), \transpose{\tilde{\phi}}) \right\}
\end{eqnarray}
where the sum runs over $\phi$ varying in an orthonormal basis of
$M_m(\C)$.
\end{lemma}


\begin{proof} By the Leibniz property, for  $V \equiv V^{\phi,
\eta} = \eta \, \transpose{\tilde{\phi}}$ one gets
\begin{eqnarray*}
\normHS{\widetilde \nabla_{e_j} V }^2 &=&  |\partial_j \eta|^2 \,
\normHS{\transpose{\tilde{\phi}}}^2  +  2 \, \eta \, \partial_j \eta \,
\inner{\transpose{\tilde{\phi}}}{\widetilde \nabla_{e_j}\transpose{\tilde{\phi}}} +
\eta^2 \, \norm{\widetilde \nabla_{e_j} \transpose{\tilde{\phi}}}^2.
\end{eqnarray*}
We notice that
$$
\sum_{\phi \in
\mathrm{ONB}}\inner{\transpose{\tilde{\phi}}}{\transpose{\tilde{\phi}}} = \sum_{\phi
\in \mathrm{ONB}} \inner{\phi}{P_U{\phi}} = \Tr P_U =
\dim{\mathcal{SU}(m)}
$$
where $\Tr$ denotes the trace in the algebra $\mathrm{End}(T_U
M_m(\C))$\footnote{\ In contrast, all over the paper we denote by $\tr$
the trace in $M_m(\C)$, \ie the ordinary matrix trace.}. Moreover,
$$
\sum_{\phi \in \mathrm{ONB}}
\inner{\transpose{\tilde{\phi}}}{\widetilde
\nabla_{e_j}\transpose{\tilde{\phi}}} = \partial_j \( \half
\sum_{\phi} \inner{\transpose{\tilde
\phi}}{\transpose{\tilde{\phi}}}\) =
\partial_j \(\half \dim{\mathcal{SU}(m)}\) = 0,
$$
so one obtains
$$
\sum_{\phi \in \mathrm{ONB}} \sum_j \normHS{\widetilde \nabla_{e_j}
V }^2 = \dim{\mathcal{SU}(m)} \, |\nabla \eta|^2 + \eta^2 \,
\sum_{\phi \in \mathrm{ONB}} \sum_j \norm{\widetilde \nabla_{e_j}
\transpose{\tilde{\phi}}}^2.
$$
By substituting in (\ref{Second variation}) and recalling that
$\dim{\mathcal{SU}(m)} = m^2 -1$ one obtains the claim.
\end{proof}

We now exploit the specific structure of $\SU$ to make the last term
in (\ref{Energy var decomposition}) more explicit. The first lemma
does not depend on the particular structure of $V^{\phi, \eta}$,
so we state it for any variational vector field $V$. Hereafter, we
set $e_j = \frac{\de}{\de k_j}$ and, in view of the embedding $\SU
\subset M_m(\C)$, we identify $dU(e_j)$ and $\de_j U$.

\begin{lemma}\label{Lemma Geometric term}
Let $V=U \psi$ with $\psi(k) \in \su$ and $U(k) \in \SU$. Then
\begin{equation}\label{Eq Geometric term}
\normHS{\widetilde \nabla_{e_j} V }^2 - R(\de_jU, V, \de_jU, V) =
\normHS{\partial_j \psi}^2 + \tr\left( U^* \de_j U \, [\psi, \,
\de_j\psi]\right).
\end{equation}
\end{lemma}

\begin{proof}
The covariant derivative in the pull-back bundle $U^*T\SU$ is equal
to the projection on $U^*T\SU$ of the ordinary derivative,
\ie
$$
\widetilde \nabla_{e_j} V(k)  = P_{U(k)} \( \frac{\de V}{\de k_j}(k)
\).
$$
By an explicit computation, and taking into account that $\psi^* = -
\psi$, one gets
$$
\widetilde \nabla_{e_j}V(k) = U(k) \(  \frac{\de \psi}{\de k_j}(k) +
\half \, [ U^*(k) {\de}_j U(k), \psi(k) ] \).
$$
Thus, one directly computes
\begin{equation}\label{Eq Cov deriv squared}
\normHS{\widetilde \nabla_{e_j} V }^2 = \normHS{\de_j\psi}^2 + \tr\(
U^* \de_j U \, \(\psi \, \de_j\psi^* - \de_j\psi \, \psi^*\)\) +
\frac{1}{4} \, \normHS{[U^* \, \de_jU, \, \psi]}^2.
\end{equation}

Since $\SU$ is a Lie group with bi-invariant metric
 $g(A,B)= \re \tr(A^*B)$ for $A, B \in \su$, the curvature tensor $R$ can
 be written in terms of Lie brackets by the \emph{Cartan formula}
\begin{equation}
\label{cartan} R(A,B,A,B)=\frac14 \tr \left( [[A,B],A]^*B \right) =
\frac14 \, \normHS{[A,B]}^2  \, .
\end{equation}
Thus, by left invariance, one has
$$
R(\de_jU, V, \de_jU, V) = R(U^* \, \de_jU, \psi, U^* \, \de_j U,
\psi) = \frac14 \, \normHS{ [ U^* \de_jU, \psi]}^2
$$
which cancels exactly the last term in (\ref{Eq Cov deriv squared}),
yielding the claim.
\end{proof}

\begin{lemma}\label{Lemma Explicit computation}
With the definitions above, one has
\begin{equation}\label{Eq explicit}
\sum_ j \sum_{\phi \in \mathrm{ONB}} \left\{ \normHS{\widetilde
\nabla_{e_j} \transpose{\tilde{\phi}}}^2 - R( \de_j U, \transpose{\tilde{\phi}},
\de_j U, \transpose{\tilde{\phi}}) \right\} = - \frac{1}{m} \sum_j
\normHS{\de_j U }^2.
\end{equation}

\end{lemma}

\begin{proof} By Lemma \ref{Lemma Geometric term} applied to $V(k) =
\transpose{\tilde{\phi}}(k)$, one immediately gets
$$
\normHS{\widetilde \nabla_{e_j} \transpose{\tilde{\phi}}}^2 - R( \de_j U,
\transpose{\tilde{\phi}}, \de_j U, \transpose{\tilde{\phi}}) = \normHS{\de_j \psi}^2
+ \tr\left( U^* \de_j U \, [\psi, \, \de_j\psi] \right)
$$
where $\psi(k) = U(k)^* \transpose{\tilde{\phi}}(k)$. By setting
$$
\hat \psi = \half (U^* \phi - \phi^* U), \qquad 
\mbox{\ie} \psi = \hat \psi - \frac{1}{m}  \tr(\hat \psi) \, \Id,
$$ one obtains $ \de_j \psi = \de_j \hat\psi - \frac{1}{m} \tr(\de_j \hat \psi) \, \Id$.
Since $\de_j\psi$ and $\Id$  are orthogonal in $M_m(\C)$,
$$
\normHS{\de_j \psi}^2  =  \normHS{\de_j \hat \psi}^2 - \frac{1}{m}
|\tr(\de_j \hat \psi)|^2.
$$

We first prove that $$\sum_{\phi \in \mathrm{ONB}} \frac{1}{m} \, |
\tr(\de_j \hat \psi)|^2 =  \frac{1}{m} \, \normHS{\de_j U}^2. $$
Indeed, by the Parseval lemma
\begin{eqnarray*}
\sum_{\phi \in \mathrm{ONB}} | \tr(\de_j \hat \psi)|^2 &=&
\sum_{\phi \in \mathrm{ONB}}  \frac{1}{4} \left| \tr\big(\de_j
U^* \phi - \phi^* \de_j U \big) \right|^2  \\
&=& \sum_{\phi \in \mathrm{ONB}} \frac{1}{4}  \left|  2 \re \tr\(
\phi^*
(-i \de_j U)\) \right|^2 \\
&=& \sum_{\phi \in \mathrm{ONB}}  |\inner{\phi}{-i \de_j U}|^2 =
\normHS{\de_j U}^2.
\end{eqnarray*}

The sum of the remaining terms vanishes, after the summation over an
ON basis. Indeed, since $[\psi, \,\de_j \psi] = [\hat \psi, \,\de_j
\hat \psi]$, one has
\begin{equation*}
\tr \left( U^* \de_j U \, [\psi, \, \de_j\psi] \right)
= \normHS{\de_j \hat \psi}^2 - \frac{1}{4} \tr \left\{ (U^* \de_jU
\phi^* U +  U^* \phi U^* \de_jU) (\de_jU^* \phi - \phi^* \de_j U)
\right\}.
\end{equation*}
Therefore,
\begin{eqnarray} \label{Quadratic in phi}
\nonumber  &&\sum_{\phi \in \mathrm{ONB}} \left \{ \normHS{\de_j
\hat \psi}^2 - \tr \left( U^* \de_j U \, [\psi, \, \de_j\psi]
\right)
\right \} \\
 &=& \sum_{\phi \in \mathrm{ONB}} \frac{1}{4} \tr \left( 2 \, U^*
\de_jU \phi^* U \de_jU^* \phi - U^* \de_jU \phi^* U \phi^* \de_jU +
U^* \phi U^* \de_jU \de_jU^* \phi \right).
\end{eqnarray}
We sum over $\phi$ in the canonical basis $\{ \Phi_{\alpha \beta},
\widetilde \Phi_{\alpha \beta} \}$,  for $\alpha, \beta \inset{1,
\ldots, m}$, where $$ \big[\Phi_{\alpha \beta} \big]_{ab}= \delta_{a
\alpha}\delta_{b \beta} \qquad \mbox{ and } \qquad 
\widetilde \Phi_{\alpha \beta} = i \Phi_{\alpha \beta}.
$$
Every term in (\ref{Quadratic in phi}) in which only $\phi$ (resp.
only $\phi^*$) appears, provides a null contribution to the sum over
$\phi$,  since the term containing $\Phi_{\alpha \beta} $ cancels
the one containing $\widetilde \Phi_{\alpha \beta}$. Therefore,
\begin{eqnarray*}
\sum_{\phi \in \mathrm{ONB}} \left \{ \normHS{\de_j \hat \psi}^2 -
\tr \left( U^* \de_j U \, [\psi, \, \de_j\psi] \right) \right \} &=&
\half \sum_{\alpha, \beta} 2 \tr\( U^* \de_jU \Phi^*_{\alpha,
\beta}  U \de_jU^* \Phi_{\alpha, \beta} \) \\
& = & | \tr(U^* \de_jU)|^2 = 0,
\end{eqnarray*}
where we used that $U$ takes values in $\SU$.
\end{proof}

By the previous lemmas, we obtain the following conclusion.

\begin{proposition}[Stability inequality]
 \label{Prop Stability inequality}
Let $U \in C^\infty(\Omega^{\prime}, \SU)$ be a local minimizer of
(\ref{Dirichletenergy}) in $\Omega^{\prime}$. Then, for every
$\Omega \subset \subset \Omega^{\prime}$ and every $\eta \in
C^{\infty}_{0}(\Omega, \R)$ one has
\begin{equation}\label{Stability ineq}
    \int_{\Omega} \, \eta^2 \, \| U^{-1} dU\|^2   \leq m (m^2
    -1) \int_{\Omega} |\nabla \eta|^2.
\end{equation}
\end{proposition}

\begin{proof} We consider the variation $U_{\epsi}^{\phi,\eta}$
corresponding to the variational vector field $V(k)= \eta(k)
\transpose{\tilde{\phi}}(k)$; by minimality one gets
$$
\left. \frac{d^2}{d\varepsilon^2} \right\vert_{\varepsilon=0} \,
E(U^{\phi, \eta}_\varepsilon;\Omega) \geq 0.
$$

By summing over $\phi$ in an ONB of $M_m(\C)$, and taking into
account Lemmas \ref{Lemma Energy Decomposition}, \ref{Lemma
Geometric term} and \ref{Lemma Explicit computation} one obtains
$$
0 \leq \sum_{\phi}  \left. \frac{d^2}{d\varepsilon^2}
\right\vert_{\varepsilon=0} \, E(U^{\phi, \eta}_\varepsilon;\Omega)
= (m^2 -1)  \int_{\Omega} |\nabla \eta|^2 - \frac{1}{m}
\int_{\Omega} \, \eta^2 \, \sum_j \normHS{\de_j U}^2,
$$
which by left invariance proves the claim.
\end{proof}


\subsection{Quantization of energy and constancy of tangent maps.}
We aim to apply the stability inequality \eqref{Stability ineq} to
degree-zero homogeneous minimizing harmonic maps (the so-called
tangent maps which appear as blow-up limits of a given minimizer
when rescaled around a given point, as in Proposition
\ref{compactness}), in order to show that they are constant.

To prove this property, we restrict ourselves to the case
$\Omega^{\prime}= \mathbb{R}^3 \setminus \{0\}$ and we assume $U$ to
be a degree-zero homogeneous harmonic map (so that the determinant
will be constant in view of Lemma \ref{Lemma SU}), i.e.  we suppose
now  that  $U(k)=\omega \left({k}/{|k|}\right)$ for some map $\omega
: S^2 \to \mathcal{U}(m)$. If $U\in W^{1,2}_{\rm loc} (\mathbb{R}^3;
\, \mathcal{U}(m))$ is a local minimizer of \eqref{Dirichletenergy},
then $\omega$ is a \emph{smooth} harmonic map $\omega \in
C^\infty(S^2; \mathcal{U}(m))$ (actually $\omega \in C^\infty(S^2;
\, \mathcal{SU}(m))$ up to multiplication by a constant unitary
matrix, in view of Lemma \ref{Lemma SU}). Indeed, for $m\geq 2$,
since $U$ is weakly harmonic and degree-zero homogeneous, then
$\omega$ is also weakly harmonic, \ie it is a critical point of the
energy functional
\begin{equation}
\label{tanenergy} \mathcal{E}(\omega)= \frac12 \int_{S^2} \left\|
\omega^{-1} d \omega \right\|^2 d\Vol \,
\end{equation}
defined on $W^{1,2}(S^2;\mathcal{U}(m))$. Since any critical point
of \eqref{tanenergy} is $C^\infty$-smooth \cite[Chapter 3]{LW}, one
concludes that $\omega$ is $C^\infty$-smooth, as claimed.

Now we take degree-zero homogeneity of $U$ into account. By
localizing the inequality \eqref{Stability ineq} on $S^2$, \ie by
taking $\eta(k)=\rho(|k|)  \psi\left({k}/{|k|}\right)$ with $\rho
\in C^{\infty}_0((0,\infty))$  and $\psi \in C^\infty(S^2)$ and
optimizing in $\rho$  (see \cite[Lemma 1.3]{SU2}), one obtains

\begin{equation}
\label{stabilitysphere} \int_{S^2}  |\psi|^2  \frac{1}{2} \left\|
\omega^{-1} d \omega \right\|^2 d\Vol \leq \frac{1}{2} m (m^2
-1) \int_{S^2} \left( |\nabla \psi|^2+\frac{1}{4} |\psi|^2 \right)
d\Vol \, .
\end{equation}

In view of the definition \eqref{tanenergy}, we set $\psi\equiv 1$
in the previous inequality and we derive the following proposition.

\begin{proposition}
\label{boundensphere} Let $m\geq 2$ and  $\omega \in C^\infty(S^2;
\mathcal{U}(m))$ an harmonic map.  If  $U(k)=\omega \left(
\frac{k}{|k|}\right)$ is a local minimizer of
\eqref{Dirichletenergy}, then $\displaystyle{ \mathcal{E}(\omega)
\leq \frac{\pi}{2}} \, m (m^2 -1)$.
\end{proposition}


In order to proceed, we recall that a very precise description of
the space of all the harmonic maps $\omega \in
C^\infty(S^2;\,\mathcal{U}(m))$ was given in \cite{U}, proving the
so-called factorization into unitons. While the latter paper
exploited algebraic techniques, a slightly different factorization
result was obtained in \cite{V}, based on an energy induction
argument.

\begin{proposition}[\cite{V}, Corollary $7^\prime$]
\label{valliUhlenbeck} Let $\omega:S^2\to \mathcal{U}(m)$ be a
nonconstant harmonic map. Then there exist a natural number $l\geq
1$ and  a canonical factorization
\begin{equation}
\label{VUfactorization} \omega=\omega_0 (p_1-p_1^\perp) \cdots
(p_l-p_l^\perp) \, , \quad 8 l\pi \leq \mathcal{E}(\omega) \, ,
\end{equation}
where $\omega_0 \in \mathcal{U}(m)$ and each $p_j$ is the hermitian
projection onto a sub-bundle of $S^2 \times \mathbb{C}^m$
holomorphic w.r.to the complex structure induced by the operator
$\bar{\partial}+ \bar{\partial} p_1+ \ldots \bar{\partial} p_{j-1}$.
\end{proposition}
Here we regard each projection $p_j$ as a map in a complex
Grassmannian $G_{k,m}(\mathbb{C})$ , $k= \mathrm{Rank}\, p_j$, and
each factor $p_j-p_j^\perp$ in \eqref{VUfactorization} as a
corresponding map into $\mathcal{U}(m)$, through the isometric
embedding $G_{k,m}(\mathbb{C}) \hookrightarrow \mathcal{U}(m)$ (the
\emph{Cartan embedding}) defined by assigning to each subspace a
unitary operator corresponding the reflection w.r.to the subspace.

The main result in \cite{V} shows that each factor in
\eqref{VUfactorization} changes the energy by an integer multiple of
$8\pi$. The consequence which will be relevant to us is the
following quantization property for a critical point of
\eqref{tanenergy}.

\begin{proposition}[\cite{V}, Corollary 8]
\label{quantchern} The energy $\mathcal{E}(\omega)$ of any harmonic
map $\omega:S^2\to \mathcal{U}(m)$ is an integer multiple of $8\pi$.
In particular, if $\omega$ is noncostant, then $\mathcal{E}(\omega) \geq
8 \pi$.
\end{proposition}

A straightforward consequence of Propositions \ref{boundensphere}
and \ref{quantchern} is that $\mathcal{E}(\omega)=0$, and hence
$\omega$ is constant (and in turn $U$ is constant) whenever
$\frac{\pi}{2} \, m (m^2 -1) < 8 \pi$, \ie for $m = 2$. As already
mentioned, this way we recover the regularity property for
minimizing harmonic maps into $S^3$ proved in \cite[Proposition
1]{SU2}, since $\mathcal{SU}(2) \equiv S^3(\sqrt{2})$.

The case $m =3$ requires additional care but the conclusion still holds, hence we have
 the following result.

\begin{corollary}
\label{liouville} Let $2\leq m\leq 3$ and let $U(k)=\omega \left(
\frac{k}{|k|}\right)$, $U\in W^{1,2}_{\rm loc}
(\mathbb{R}^3;\mathcal{U}(m))$ be a local minimizer of
\eqref{Dirichletenergy}.  Then $U$ is constant.
\end{corollary}

\begin{proof}
As already recalled at the beginning of this subsection, $\omega \in
C^\infty (S^2; \mathcal{U}(m))$. Since the case $m=2$ follows
readily from Proposition \ref{boundensphere}, we assume $m=3$, and,
arguing by contradiction, we may assume $U \not \equiv
\mathrm{const}$. By Proposition \ref{boundensphere} we obtain
$\mathcal{E}(\omega)\leq 12 \pi$, hence $\mathcal{E}(\omega)=8\pi$
because of Proposition \ref{quantchern}. Going back to Proposition
\ref{valliUhlenbeck} we see that, up to a constant unitary matrix,
we may assume  that the product in \eqref{VUfactorization} contains
only one factor. More precisely, since $S^2=\mathbb{C}P^1$,
$G_{1,3}(\mathbb{C})=\mathbb{C}P^2$ and, up to isometry,
$G_{1,3}(\mathbb{C})=G_{2,3}(\mathbb{C})$, then $\omega=p-p^\perp$,
where $p:\mathbb{C}P^1 \to \mathbb{C}P^2$  is an holomorphic map
(equivalently, an holomorphic line bundle over $\mathbb{C}P^1$), and
$\mathcal{E}(\omega)=8\pi$. \newline Now, as in \cite[page 131]{V},
let $\tilde{\omega}$ be the K\"ahler form on $\mathbb{C}P^2$
normalized to be the positive generator in
$H^2(\mathbb{C}P^2;\mathbb{Z})$. Then, in homogeneous coordinates
$[z_1,z_2] \in \mathbb{C}P^1$, $p$ is a polinomial map of degreee
one, since $8\pi |\deg \, p|=|\int_{\mathbb{C}P^1} p^*
\tilde{\omega}|=8\pi |c_1 ( p )|=\mathcal{E}( \omega)=8\pi$, where
$c_1 ( p )$ is the first Chern number of the holomorphic bundle $p$.
As $p$ is nonconstant, $p(z_1,z_2)=[p_1(z_1,z_2), \, p_2(z_1,z_2),
\, p_3(z_1,z_2)]\in \mathbb{C}P^2$ for three suitable (not all
proportional) degree-one polynomials. Since the degree is one, they
are linearly dependent, hence, up to a linear (resp. unitary) change
of coordinates on $\mathbb{C}P^2$ (resp. on $\mathbb{C}^3$), we may
assume $p_3 \equiv 0$.  To summarize, we can regard $p:\mathbb{C}P^1
\to \mathbb{C}P^1 \subset \mathbb{C}P^2$ (the range of $p$ being the
set of horizontal lines in $\mathbb{C}^3$) and, up to a further
reflection in the third coordinate of $\mathbb{C}^3$, $\omega: S^2
\to \mathcal{U}(2) \subset \mathcal{U}(3)$ through the diagonal
embedding (the unitary operator being now the identity on the third
coordinate). As we are back to the case $m=2$, \ie $U$ is a
$\mathcal{U}(2)$-valued local minimizer, we conclude $U\equiv
\mathrm{const}$, and the contradiction concludes the proof.
\end{proof}

Combining the previous corollary with the well known regularity theory for harmonic maps in three dimension (see \cite{SU1}, \cite{S} and \cite{LW}) we have the following result which seems to be of independent interest.

\begin{theorem}
\label{RegUnitons} Let $2\leq m \leq 3$ and $\Omega \subset
\mathbb{R}^3$ an open set. If $U \in
W^{1,2}_{\rm loc} (\Omega;\mathcal{U}(m))$ is a local minimizer of
\eqref{Dirichletenergy} then $U$ is real-analytic.
\end{theorem}



\begin{thebibliography}{9999999}

\bibitem[Bl]{Bl} \textsc{Blount, E. I.} : Formalism of Band Theory. In :
Seitz, F., Turnbull, D. (eds.),  {\it Solid State Physics} {\bf 13},
pages 305--373, Academic Press, 1962.

\bibitem[BZ]{BZ} \textsc{Bethuel, F.;  Zheng, X.M.} : Density of smooth functions between
two manifolds in Sobolev spaces. {\it J. Funct. Anal.} {\bf 80}
(1988),  60--75.

\bibitem[BPCM]{Wannier_letter_BPCM}\textsc{Brouder Ch.; Panati G.; Calandra M.; Mourougane Ch.;
Marzari N.}:  Exponential localization of Wannier functions in
insulators, {\it Phys. Rev. Lett.} \textbf{98} (2007), 046402.


\bibitem[Ca]{Ca} \textsc{Campanato, S.} : Propriet\`a di h\"olderianit\`a di alcune classi
di funzioni.  {\it Ann. Scuola Norm. Sup. Cl. Sci.}  {\bf 17}
(1963), 175--188.

\bibitem[CaLeLi]{CattoLeBrisLions} \textsc{Catto I.; Le Bris C.;
Lions P.-L.}: On the thermodynamic limit for Hartree-Fock type
problems, {\it Ann. I. H. Poincaré} {\bf 18} (2001), 687-760.

\bibitem[CaDeLe]{CancesDelaurenceLewin} \textsc{Cancès E.; Deleurence A.;
Lewin M.} : A New Approach to the Modeling of Local Defects in
Crystals: The Reduced Hartree-Fock Case, {\it Commun. Math. Phys.}
\textbf{281} (2008), 129--177.


\bibitem[CWY]{CWY} \textsc{Chang, S. Y.; Wang, L.; Yang, P.} :  Regularity of harmonic maps.
 {\it Comm. Pure Appl. Math.}  {\bf 52} (1999),  1099--1111.

\bibitem[Cl1]{Cl1} \textsc{des Cloizeaux, J.} : Energy bands and projection operators
in a crystal: Analytic and asymptotic properties. {\it  Phys. Rev. } {\bf  135} (1964),  A685--A697.

\bibitem[Cl2]{Cl2} \textsc{des Cloizeaux, J.} : Analytical properties of n-dimensional energy
bands and Wannier functions. {\it Phys. Rev.} {\bf 135} (1964),
A698--A707.

\bibitem[CLMS]{CLMS} \textsc{Coifman, R.; Lions, P.-L.; Meyer, Y.; Semmes, S.} :
Compensated compactness and Hardy spaces. {\it J. Math. Pures Appl.} {\bf 72} (1993),
 247--286.

\bibitem[CNN]{CoNe} \textsc{Cornean, H.D.; Nenciu A.; Nenciu, G.} : Optimally localized Wannier functions
for quasi one-dimensional nonperiodic insulators. {\it J. Phys. A:
Math. Theor.} {\bf 41} (2008), 125202.

\bibitem[DuNo]{DuNo} \textsc{B. A. Dubrovin, S. P. Novikov} : Ground state of a two-dimensional
electron in a periodic magnetic field, {\it Zh. Eksp. Teor. Fiz}
{\bf 79} (1980), 1006-1016. Translated in {\it Sov. Phys. JETP} {\bf
52}, vol. 3 (1980) 511--516.

\bibitem[Go]{Goedecker} \textsc{Goedecker, S.} :  Linear scaling electronic structure
methods, {\it Rev. Mod. Phys.} \textbf{71} (1999), 1085--1111.

\bibitem[Ha]{Haldane88} \textsc{Haldane, F.D.M.} :  Model for a Quantum Hall
effect without Landau levels: condensed-matter realization of the
\virg{parity anomaly}. {\it Phys. Rev. Lett.} {\bf 61} (1988), 2017.

\bibitem[HL]{HL} \textsc{Hang, F.; Lin, F.H.} : Topology of Sobolev mappings. II.
{\it Acta Math.} {\bf 191} (2003),  55--107.

\bibitem[HS]{HeSj_Scroedinger} \textsc{Helffer, B.;  Sj\"ostrand, J.} :
\'Equation de Schr\"odinger avec champ magnétique et équation de
Harper. In: {Schr\"odinger operators}, Lecture Notes in Physics 345,
Springer, Berlin, 1989, 118--197.

\bibitem[HW]{HW} \textsc{Howard, R.; Wei, S. W.} : Nonexistence of stable harmonic maps to and from certain homogeneous spaces and submanifolds of Euclidean space. {\it Trans. Amer. Math. Soc.} {\bf 294} (1986), 319-331.
\bibitem[Jo]{J} \textsc{Jost, J.} :  Riemannian geometry and geometric analysis.
Universitext. Springer, Berlin, 1995.

\bibitem[Ka]{Kato} \textsc{Kato, T.} : Perturbation theory for
linear operators. Springer, Berlin, 1966.

\bibitem[Ki]{Kievelsen} \textsc{Kievelsen, S.} : Wannier functions in
one-dimensional disordered systems: application to fractionally
charged solitons, {\it Phys. Rev. B} {\bf  26}, 4269--4274 (1982).

\bibitem[KSV]{KSV} \textsc{King-Smith,  R.~D. ; Vanderbilt, D.} : Theory of polarization of crystalline solids,
{\it Phys. Rev. B} {\bf  47}, 1651--1654 (1993).

\bibitem[Ko]{Kohn64} \textsc{Kohn, W.} :  Analytic Properties of Bloch Waves and Wannier Functions.
Phys. Rev. {\bf 115}, 809 (1959)


\bibitem[Ku1]{Kuchment_book} \textsc{Kuchment, P.} : Floquet Theory for Partial Differential Equations.
Operator Theory: Advances and Applications. Birkhäuser, 1993


\bibitem[Ku2]{Kuchment} \textsc{Kuchment, P.} : Tight frames of exponentially decaying Wannier
functions, {\it J. Phys. A: Math. Theor.} {\bf 42} (2009), 025203.

\bibitem[Kui]{Kui}\textsc{Kuiper, N.} :  The homotopy type of the unitary group of Hilbert space.
{\it Topology} {\bf 3} (1965) 19--30.

\bibitem[LW1]{LW} \textsc{Lin, F.H.; Wang, C.Y.} :  The analysis of harmonic maps and their heat flows.
World Scientific, 2008.

\bibitem[LW2]{LW2} \textsc{Lin, F.H.; Wang, C.Y.}
: Stable stationary harmonic maps to spheres. {\it Acta Math. Sin.
(Engl. Ser.) }  {\bf 22} (2006), 319--330.


\bibitem[MaVa]{MaVa} \textsc{Marzari, N. ; Vanderbilt, D.}:
Maximally localized generalized Wannier functions for composite
energy bands. {\it  Phys. Rev. B}   \, {\bf 56}  (1997),
12847--12865.

\bibitem[MSV]{Lyon conference} \textsc{Marzari, N.; Souza, I.; Vanderbilt, D.}
An introduction to maximally localized Wannier functions. Highlight
of the Month, \emph{Psi-K Newsletter} \textbf{57} (2003), 129--168.

\bibitem[MYSV]{Wannier review} \textsc{Marzari, N.;  Mostofi A.A.; Yates J.R.;
Souza I.; Vanderbilt D.} Maximally localized Wannier functions:
Theory and applications, \emph{Rev. Mod. Phys.} \textbf{84}, 1419 (2012).


\bibitem[Mor]{Morrey} \textsc{Morrey, C.B.} : Multiple integrals in the calculus of variations.
Grund. der math. Wissenschaften 130. Springer, New York, 1966.

\bibitem[Mos]{Mos} \textsc{Moser, R.} : Partial regularity for harmonic maps and
related problems. World Scientific, 2005.

\bibitem[NaMa]{NaMa}\textsc{Lee, Y.-S.; Nardelli, M. B.; Marzari, N.} :
Band structure and quantum conductance of nanostructures from
maximally localized Wannier functions: the case of functionalized
carbon nanotubes {\it Phys. Rev. Lett. }{\bf 95} (2005), 076804.


\bibitem[Ne1]{Ne83} \textsc{Nenciu, G.} : Existence of the exponentially localised
Wannier functions. {\it Comm. Math. Phys.} \, {\bf 91} (1983),
81--85.

\bibitem[Ne2]{Ne91} \textsc{Nenciu, G.} : Dynamics of band electrons in electric and magnetic
fields: Rigorous justification of the effective Hamiltonians. {\it
Rev. Mod. Phys.} {\bf 63} (1991), 91--127.

\bibitem[NN]{NeNe97}\textsc{Nenciu, A.;  Nenciu, G.} : The existence of generalized Wannier
functions for one-dimensional systems.  {\it Comm. Math. Phys.}
{\bf 190} (1988), 541--548.

\bibitem[No]{No} \textsc{Novikov, S.P.} : Magnetic Bloch functions and vector bundles.
Typical dispersion law and quantum numbers. {\it Sov. Math. Dokl.}
{\bf 23} (1981),  298--303.

\bibitem[PST]{PST2003} \textsc{Panati, G.; Spohn, H.; Teufel, S.} : Effective dynamics for
Bloch electrons: Peierls substitution and beyond. {\it Comm. Math.
Phys.} {\bf 242} (2003), 547--578.

\bibitem[Pa]{Pa2008} \textsc{Panati, G.}: Triviality of Bloch and Bloch-Dirac bundles.
{\it Ann. Henri Poincar\'e} \, {\bf 8} (2007),  995--1011.

\bibitem[Pr]{Prodan11} \textsc{Prodan, E.} : Disordered topological insulators:
a Non-Commutative Geometry perspective. {\it J. Phys. A: Math.
Theor. } {\bf 44} (2011), 113001.

\bibitem[RS]{RS4} \textsc{Reed M., Simon, B.} :  Methods of Modern Mathematical Physics.
Volume IV: Analysis of Operators. Academic Press, New York, 1978.

\bibitem[Re]{Re92} \textsc{Resta, R.} : Theory of the electric polarization in crystals,
\textit{Ferroelectrics} \textbf{136}, 51--75 (1992).

\bibitem[Si]{S} \textsc{Simon, L.} : Theorems on regularity and singularity of energy
minimizing maps.  Lectures in Mathematics ETH Z\"urich.
Birkh\"auser, Basel,  1996.

\bibitem[St]{St} \textsc{Stein, E.} :  Harmonic analysis: real-variable methods, orthogonality,
 oscillatory integrals. Princeton Mathematical Series, 43. Princeton University Press,
 Princeton, 1993.

\bibitem[SU1]{SU1}\textsc{Schoen, R.; Uhlenbeck, K.} : A regularity theory for harmonic maps.
{\it J. Differential Geom.} {\bf 17} (1982),  307--335.

\bibitem[SU2]{SU2} \textsc{Schoen, R.; Uhlenbeck, K.} : Regularity
of minimizing harmonic maps into the sphere. {\it Invent. Math.}  {\bf 78} (1984),  89--100.

\bibitem[TCVR]{ThoCerVanRes2005} \textsc{Thonhauser, T.; Ceresoli, D.; Vanderbilt, D.; Resta,
R.} : Orbital magnetization in periodic insulators, {\it Phys. Rev.
Letters } \textbf{95} (2005), 137205.

\bibitem[Uh]{U} \textsc{Uhlenbeck, K.}:  Harmonic maps into Lie groups: classical
solutions of the chiral model. {\it J. Differential Geom.} {\bf 30} (1989),  1--50.

\bibitem[Va]{V} \textsc{Valli, G.}: On the energy spectrum of harmonic 2-spheres
in unitary groups. {\it Topology} {\bf  27} (1988),  129--136.


\bibitem[Wa]{Wannier} \textsc{Wannier, G. H.} : The structure of electronic excitation
levels in insulating crystals. {\it Phys. Rev. }{\bf 52} (1937),
191--197.

\bibitem[We]{We} \textsc{Wei, W.}: Liouville Theorems for stable harmonic maps into either strongly unstable, or $\delta-$pinched manifolds. \emph{Proceedings of Symposia in Pure Mathematic} {\bf 44} (1986).

\bibitem[Xin]{X} \textsc{Xin, Y.L.}: Geometry of harmonic maps. Progress in Nonlinear Differential Equations and their Applications, 23. Birkh\"{a}user,
Boston, 1996.


\end{thebibliography}
\end{document}